\documentclass[11pt]{article}

\usepackage[colorlinks, linkcolor=blue, citecolor=blue]{hyperref}            
\usepackage{color}
\usepackage{graphicx,subfigure,amsmath,amssymb,amsfonts,bm,epsfig,epsf,url,dsfont}
\usepackage{amsthm}
\usepackage{tikz}
\usepackage{bbm}      
\usepackage{booktabs}
\usepackage{cases}
\usepackage{fullpage}
\usepackage[small,bf]{caption}
\usepackage[backend=biber,style=alphabetic,maxbibnames=99]{biblatex}
\usepackage[top=1in,bottom=1in,left=1in,right=1in]{geometry}
\usepackage{fancybox}
\usepackage[bottom]{footmisc}
\usepackage{appendix}

\usepackage{algorithmic}
\usepackage{amsmath}
\usepackage{amssymb}
\usepackage{color}
\usepackage{mathtools}
\usepackage{makecell}
\usepackage{bm}
\usepackage{todonotes}
\usepackage{diagbox}
\usepackage{tablefootnote}

\DeclarePairedDelimiter{\ceil}{\lceil}{\rceil}

\DeclareMathOperator{\poly}{poly}
\DeclareMathOperator{\polylog}{polylog}

\DeclareMathOperator{\Ber}{Ber}

\newcommand{\latestedits}[1]{{#1}}

\let\baraccent=\= %
\renewcommand{\=}[1]{\stackrel{#1}{=}} %

\providecommand{\RR}{\mathbb{R}}

\providecommand{\CC}{\mathbb{C}}
\providecommand{\FF}{\mathbb{F}}

\providecommand{\ZZ}{\mathbb{Z}}

\providecommand{\cA}{\mathcal{A}}

\providecommand{\cL}{\mathcal{L}}

\providecommand{\PP}{\mathbb{P}}

\providecommand{\eps}{\epsilon}

\mathchardef\mhyphen="2D %

\providecommand{\sm}{\setminus}

\newcommand{\interior}[1]{%
  {\kern0pt#1}^{\mathrm{o}}%
}

\newtheorem{theorem}{Theorem}[section]
\newtheorem{lemma}[theorem]{Lemma}
\newtheorem{definition}[theorem]{Definition}

\newtheorem{conjecture}[theorem]{Conjecture}
\newtheorem{claim}[theorem]{Claim}
\newtheorem{proposition}[theorem]{Proposition}
\newtheorem{corollary}[theorem]{Corollary}

\newenvironment{fminipage}%
  {\begin{Sbox}\begin{minipage}}%
  {\end{minipage}\end{Sbox}\fbox{\TheSbox}}

\newenvironment{algbox}[0]{\vskip 0.2in
\noindent 
\begin{fminipage}{6.3in}
}{
\end{fminipage}
\vskip 0.2in
}

\providecommand{\cl}{\mathrm{cl}}

\newcommand{\dtv}{d_{\text{TV}}}

\title{The Average-Case Complexity of Counting Cliques \\ in Erd\H{o}s-R\'enyi Hypergraphs}
\author{Enric Boix{-}Adser\`a\thanks{Massachusetts Institute of Technology. Department of EECS. Email: \texttt{eboix@mit.edu}.}
\and
Matthew Brennan\thanks{Massachusetts Institute of Technology. Department of EECS. Email: \texttt{brennanm@mit.edu}.}
\and 
Guy Bresler\thanks{Massachusetts Institute of Technology. Department of EECS. Email: \texttt{guy@mit.edu}.}}
\date{\today}

\begin{document}

\makeatletter
\def\@maketitle{%
  \newpage
  \null
  \vskip 2em%
  \begin{center}%
  \let \footnote \thanks
    {\LARGE \@title \par}%
    \vskip 1.5em%
    {\large
      \lineskip .5em%
      \begin{tabular}[t]{c}%
        \@author
      \end{tabular}\par}%
    \vskip 1em%
    {\large \@date}%
    \vskip 1em
    {\textit{\normalsize Dedicated to the memory of our dear colleague and friend, Matthew Brennan}}
  \end{center}%
  \par
  \vskip 0.5em}
\makeatother

\maketitle

\pagenumbering{roman}

\begin{abstract}
The complexity of clique problems on Erd\H{o}s-R\'{e}nyi random graphs has become a central topic in average-case complexity. Algorithmic phase transitions in these problems have been shown to have broad connections ranging from mixing of Markov chains and statistical physics to information-computation gaps in high-dimensional statistics.
We consider the problem of counting $k$-cliques in $s$-uniform Erd\H{o}s-R\'{e}nyi hypergraphs $G(n, c, s)$ with edge density $c$ and show that its fine-grained average-case complexity can be based on its worst-case complexity. We give a worst-case to average-case reduction for counting $k$-cliques on worst-case hypergraphs given a blackbox solving the problem on $G(n, c, s)$ with low error probability. Our approach is closely related to [Goldreich and Rothblum, FOCS18], which showed a worst-case to average-case reduction for counting cliques for an efficiently-sampleable distribution on graphs.

Our reduction has the following implications:
\begin{itemize}
\item \textit{Dense Erd\H{o}s-R\'{e}nyi graphs and hypergraphs}: Counting $k$-cliques on $G(n, c, s)$ with $k$ and $c$ constant matches its worst-case complexity up to a $\polylog(n)$ factor. Assuming randomized ETH, it takes $n^{\Omega(k)}$ time to count $k$-cliques in $G(n, c, s)$ if $k$ and $c$ are constant.
\item \textit{Sparse Erd\H{o}s-R\'{e}nyi graphs and hypergraphs}: When $c = \Theta(n^{-\alpha})$, we give several algorithms exploiting the sparsity of $G(n, c, s)$ that are faster than the best known worst-case algorithms. Complementing this, based on a fine-grained worst-case assumption, our reduction implies a different average-case phase diagram for each fixed $\alpha$ depicting a tradeoff between a runtime lower bound and $k$.
Surprisingly, in the hypergraph case ($s \ge 3$), these lower bounds are tight against our algorithms exactly when $c$ is above the Erd\H{o}s-R\'{e}nyi $k$-clique percolation threshold.
\end{itemize}
Our reduction is the first worst-case to average-case reduction for a problem over Erd\H{o}s-R\'{e}nyi hypergraphs and is the first mapping from a worst-case problem to an average-case problem with a different fine-grained complexity that we are aware of.
We also give a variant of our worst-case to average-case reduction for computing the parity of the $k$-clique count that requires a milder assumption on the error probability of the blackbox solving the problem on $G(n, c, s)$.
\end{abstract}

\pagebreak

\pagebreak

\pagenumbering{arabic}

\section{Introduction}
We consider the average-case complexity of counting $k$-cliques in $s$-uniform Erd\H{o}s-R\'{e}nyi hypergraphs $G(n, c, s)$, where every $s$-subset of the $n$ vertices is a hyperedge independently with probability $c$. Our main result is a reduction for counting $k$-cliques on worst-case hypergraphs given a blackbox algorithm solving the problem on $G(n, c, s)$ with low error probability.  Our approach is closely related to the recent work  \cite{goldreich2018counting}, which showed a worst-case to average-case reduction for counting cliques for a particular efficiently-samplable distribution on graphs.
Our reduction yields two different sets of average-case lower bounds for counting $k$-cliques in graphs sampled from the natural distribution $G(n, c, s)$ in the dense and sparse cases of $c = \Theta(1)$ and $c = \Theta(n^{-\alpha})$, with tradeoffs between runtime and $c$.
We also show that these average-case lower bounds often match algorithmic upper bounds.

The complexity of clique problems on Erd\H{o}s-R\'{e}nyi random graphs has become a central topic in average-case complexity, discrete probability and high-dimensional statistics. A body of work has analyzed algorithms for finding large cliques in Erd\H{o}s-R\'{e}nyi graphs\footnote{In both ordinary Erd\H{o}s-R\'{e}nyi graphs and the planted clique model.} \cite{kuvcera1995expected,alon1998finding, feige2000finding, mcsherry2001spectral, feige2010finding, ames2011nuclear, dekel2014finding, deshpande2015finding, chen2016statistical}, and hardness results have been shown for 
greedy algorithms \cite{karp1976, grimmett1975colouring, jerrum1992large, mcdiarmid1984colouring, pittel1982probable}, local algorithms \cite{gamarnik2014limits, coja2015independent, rahman2017local}, query models \cite{feige2018finding}, bounded-depth circuits \cite{rossman2008constant}, monotone circuits \cite{rossman2010monotone}, low-degree sum of squares (SOS) relaxations \cite{barak2016nearly}, statistical query algorithms \cite{feldman2013statistical}, and resolution \cite{atserias2018clique}. The hardness of clique problems on Erd\H{o}s-R\'enyi graphs has been used as an average-case assumption in cryptography \cite{juels2000hiding} and to show information-computation gaps in a variety of statistical problems \cite{berthet2013complexity, koiran2014hidden, chen2015incoherence, hajek2015computational, ma2015computational, brennan2018reducibility, brennan2019universality, brennan2019optimal}.

All of the above lower bounds for clique problems on Erd\H{o}s-R\'{e}nyi random graphs are against restricted classes of algorithms. One reason for this is that there are general obstacles to basing average-case complexity on worst-case complexity. For example, natural approaches to polynomial-time worst-case to average-case reductions for NP-complete problems fail unless coNP $\subseteq$ NP/poly \cite{feigenbaum1993random, bogdanov2006worst, bogdanov2006average}. The objective of this work is to show that this worst-case characterization of average-case complexity is possible in a fine-grained sense for the natural problem of counting $k$-cliques in $s$-uniform Erd\H{o}s-R\'{e}nyi hypergraphs $G(n, c, s)$ with edge density $c$. 

A motivating recent work by Goldreich and Rothblum \cite{goldreich2018counting} also considered worst-case to average-case reductions for $k$-clique counting. They provided such a reduction mapping to an efficiently sampleable distribution on graphs with a high min-entropy of $\tilde{\Omega}(n^2)$. In contrast to \cite{goldreich2018counting}, our objectives are to: (1) map precisely to the natural distribution $G(n, c, s)$ for different edge densities $c$, including $c = \Theta(1)$ and the sparse case $c = \Theta(n^{-\alpha})$; and (2) to characterize the tradeoff between the time-complexity of counting $k$-cliques in $G(n, c, s)$ and the sparsity parameter $\alpha$. Achieving this requires new ingredients for the self-reducibility of counting $k$-cliques as a low-degree polynomial and a tight analysis of random biased binary expansions over $\mathbb{F}_p$ with finite Fourier analysis. 

However, our techniques also come at the cost of requiring a low error probability ($1/\text{polylog}(n)$ in the dense case and $1/\poly(n)$ in the sparse case) for the average-case blackbox solving $k$-clique counting on $G(n, c, s)$. This is in contrast to \cite{goldreich2018counting}, where a very high error probability of $1 - 1/\text{polylog}(n)$  is tolerated. It remains an interesting open problem to extend our results for $G(n, c, s)$ to tolerate higher error blackboxes. This error tolerance and open problem are discussed further in Sections \ref{subsec:reductionthmstatements} and \ref{sec:openproblems}, and how our techniques relate to those in \cite{goldreich2018counting} is discussed in Sections \ref{subsec:wcacoverview} and \ref{sec:averagecasereductionproof}. As a step towards increasing the allowed blackbox error, we also give a variant of our reduction for computing the \emph{parity} of the $k$-clique count that only requires a constant bound on the error probability (for each fixed $k$) of the blackbox algorithm solving the problem on $G(n, c, s)$ when $c = 1/2$. We now give an overview of our contributions.

\subsection{Overview of Main Results}

We provide two complementary main results on the fine-grained average-case complexity of counting $k$-cliques in $G(n, c, s)$. The precise formulations of the problems we consider are in Section \ref{sec:worstcasehardnessconjectures}.

\paragraph{Worst-case to average-case reduction} We give a worst-case to average-case reduction from counting $k$-cliques in worst-case $s$-uniform hypergraphs to counting $k$-cliques in hypergraphs drawn from $G(n, c, s)$. The key guarantees of this reduction are summarized in the following simplified version of our main theorem.

\begin{theorem}[Simplified Main Result]
If $2 \leq s \le k$ are constant integers and $c = c(n)$ satisfies $0 < c \le 1 - \Omega(1)$, then there is a parameter $\Upsilon_{\#} = c^{-\binom{k}{s}} (\log n)^{O(1)}$ such that the following holds. If there is a randomized algorithm counting $k$-cliques in time $O(n^t)$ with error probability less than $1/\Upsilon_{\#}$ on hypergraphs drawn from $G(n, c, s)$, then there is a randomized algorithm counting $k$-cliques on worst-case $s$-uniform hypergraphs with error probability less than $1/3$ running in time $O\left(\Upsilon_{\#} \cdot n^{\max\{t, s\}} \right)$.
\end{theorem}

We discuss the necessity of the error tolerance and the multiplicative slowdown in our worst-case to average-case reduction in Section \ref{subsec:reductionthmstatements}. This result has a number of consequences for basing the average-case fine-grained complexity of $k$-clique counting over Erd\H{o}s-R\'{e}nyi hypergraphs on its worst-case complexity, which we now overview.

Counting $k$-cliques in worst-case hypergraphs is known to take $n^{\Omega(k)}$ time for randomized algorithms assuming the randomized Exponential Time Hypothesis (rETH)\footnote{rETH asserts that any randomized algorithm takes at least $2^{c n}$ time to solve 3-SAT in the worst-case, for some constant $c > 0$.} if $k$ does not grow with $n$ \cite{chen2006strong,calabro2008complexity}. The best known worst-case algorithms up to subpolynomial factors are the $O\left(n^{\omega \lceil k/3 \rceil}\right)$ time algorithm of \cite{nevsetvril1985complexity} in the graph case of $s = 2$ and exhaustive $O(n^k)$ time search on worst-case hypergraphs with $s \ge 3$. Here, $\omega \leq 2.373$ denotes the matrix multiplication constant. Our reduction is the first worst-case to average-case reduction to Erd\H{o}s-R\'{e}nyi hypergraphs. It has different implications for the cases of dense and sparse hypergraphs because of the factor $\Upsilon_{\#}$, as described next.
\begin{enumerate}
\item \textit{Dense Erd\H{o}s-R\'{e}nyi graphs and hypergraphs.} When $k$ and $c$ are constant, our reduction constructs an efficient $k$-clique counting algorithm that succeeds on a worst-case input hypergraph with high probability, using $\polylog(n)$ queries to an average-case oracle that correctly counts $k$-cliques on a $1 - 1/\polylog(n)$ fraction of Erd\H{o}s-R\'enyi hypergraphs drawn from $G(n,c,s)$. This essentially shows that $k$-clique counting in the worst-case matches that on dense Erd\H{o}s-R\'{e}nyi hypergraphs. More precisely, $k$-clique counting on $G(n, c, s)$ with $k, c$ and $s$ constant must take $\tilde{\Omega}\left( n^{\omega \lfloor k/3 \rfloor} \right)$ time when $s = 2$ and $\tilde{\Omega}(n^k)$ time when $s \ge 3$, unless there are faster worst-case algorithms. Furthermore, our reduction shows that it is rETH-hard to count $k$-cliques in $n^{o(k)}$ time on $G(n, c, s)$ with $k, c$ and $s$ constant.
\item \textit{Sparse Erd\H{o}s-R\'{e}nyi graphs and hypergraphs.} Our reduction also applies with a different multiplicative slowdown and error tolerance to the sparse case of $c = \Theta(n^{-\alpha})$, where the fine-grained complexity of $k$-clique counting on $G(n, c, s)$ is very different than on worst-case inputs. Our reduction implies fine-grained lower bounds of $\tilde{\Omega}\left(n^{\omega \lceil k/3 \rceil - \alpha \binom{k}{2}} \right)$ when $s = 2$ and $\tilde{\Omega}\left(n^{k - \alpha \binom{k}{s}} \right)$ when $s \ge 3$ for inputs drawn from $G(n, c, s)$, unless there are faster worst-case algorithms. We remark that in the hypergraph case of $s \ge 3$, this lower bound matches the expectation of the quantity being counted, the number of $k$-cliques in $G(n, c, s)$, up to $\polylog(n)$ factors.\footnote{For the sub-class of algorithms that enumerate $k$-cliques one by one, the $k$-clique count is a trivial lower bound on the runtime. Our general lower bound matches this heuristic lower bound.}
\end{enumerate}
Precise statements of our results can be found in Section \ref{subsec:reductionthmstatements}. For simplicity, our results should be interpreted as applying to algorithms that succeed with probability $1 - (\log n)^{-\omega(1)}$ in the dense case and $1 - n^{-\omega(1)}$ in the sparse case.

We also give a second worst-case to average-case reduction for computing the parity of the number of $k$-cliques which has a weaker requirement of $1 - \Theta_{k, s}(1)$ on the error probability for the blackbox solving the problem on $G(n, c, s)$ in the dense case of $c = 1/2$. We provide an overview of our multi-step worst-case to average-case reduction in Section \ref{subsec:wcacoverview}. The steps are described in detail in Section \ref{sec:averagecasereductionproof}.

\paragraph{Algorithms for $k$-clique counting on $G(n, c, s)$} We also analyze several natural algorithms for counting $k$-cliques in sparse Erd\H{o}s-R\'{e}nyi hypergraphs. These include an extension of the natural greedy algorithm mentioned previously from $k$-\textsc{clique} to counting $k$-cliques, a modification to this algorithm using the matrix multiplication step of \cite{nevsetvril1985complexity} and an iterative algorithm achieving nearly identical guarantees. These algorithms count $k$-cliques in $G(n, c, s)$ when $c = \Theta(n^{-\alpha})$ with several different runtimes, the best of which are as follows:
\begin{itemize}
\item $\tilde{O}\left( n^{k + 1 - \alpha \binom{k}{s}} \right)$ if $s \ge 3$ and $k < \tau + 1$;
\item $\tilde{O}\left( n^{\tau + 2 - \alpha \binom{\tau + 1}{s}} \right)$ if $s \ge 3$ and $\tau + 1 \le k \le \kappa + 1$; and
\item $\tilde{O}\left( n^{\omega \lceil k/3 \rceil + \omega - \omega \alpha \binom{\lceil k/3 \rceil}{2}} \right)$ if $s = 2$ and $k \le \kappa + 1$.
\end{itemize}
Here, $\tau$ and $\kappa$ are the largest positive integers satisfying that $\alpha \binom{\tau}{s - 1} < 1$ and $\alpha \binom{\kappa}{s - 1} < s$. The thresholds $\kappa$ and $\tau$ have natural interpretations as roughly the clique number and most frequent clique size in the graph $G(n, c, s)$, respectively. Throughout, we restrict our attention to $k$ with $k \le \kappa + 1$ since the probability that the largest clique in $G$ has size $\omega(G) > \kappa + 1$ is $1/\text{poly}(n)$.

\latestedits{The threshold $\tau + 1$ also has a natural interpretation as the $k$-clique percolation threshold \cite{derenyi2005clique, palla2007critical, dorogovtsev2008critical}, defined below. Given a hypergraph $G$, define two $k$-cliques of $G$ to be adjacent if they share $(k - 1)$ of their $k$ vertices. This induces a hypergraph $G_k$ on the set of $k$-cliques. For graphs $G$ drawn from $G(n, c)$, \cite{derenyi2005clique} introduced the $k$-clique percolation threshold of $c = \frac{1}{k - 1} \cdot n^{-\frac{1}{k - 1}}$, above which a giant component emerges in $G_k$. This threshold and extensions were rigorously established in \cite{bollobas2009clique}. In the graph case of $s=2$, this threshold matches $\tau + 1$, which is the largest integer $k$ such that $\alpha < \frac{1}{k - 1}$. Following the same heuristic as in \cite{derenyi2005clique}, our threshold $\tau + 1$ is a natural extension of the $k$-clique percolation threshold to the hypergraph case of $s \ge 3$. In other words, $\tau + 1$ roughly corresponds to the largest value of $k$ at which a local search algorithm can explore all the cliques in the hypergraph starting from any given clique.}

\begin{figure*}[t!]
\centering
\begin{tikzpicture}[scale=0.065]
\tikzstyle{every node}=[font=\footnotesize]
\def\xmin{0}
\def\xmax{108}
\def\ymin{0}
\def\ymax{60}

\fill [green!20, domain=0:100, variable=\x]
(0, 0)
-- plot ({\x}, {2.373*\x/3 - 2.373*(2/99)*(\x/3)*(\x/3 - 1)/2})
-- (100, 58) -- (0, 58) -- (0, 0);

\fill [gray!20, domain=0:100, variable=\x]
(0, 0)
-- plot ({\x}, {2.373*\x/3 - 2.373*(2/99)*(\x/3)*(\x/3 - 1)/2})
-- (100, 0) -- (0, 0);

\fill [blue!20, domain=0:79.35, variable=\x]
(0, 0)
-- plot ({\x}, {2.373*\x/3 - (2/99)*(\x)*(\x - 1)/2})
-- (0, 0);

\node at (20, 60) [below] {Graphs ($s = 2$)};
\node at (12, 40) {feasible};
\node at (40, 5) {infeasible};
\node at (90, 40) {open};
\node at (70, 15) {$\frac{\omega k}{3} - \alpha \binom{k}{2}$};
\node [rotate=28] at (50, 38) {$\frac{\omega k}{3} - \frac{\omega \alpha}{9} \binom{k}{2}$};

\draw[->] (\xmin,\ymin) -- (\xmax,\ymin) node[right] {$k$};
\draw[->] (\xmin,\ymin) -- (\xmin,\ymax) node[above] {$\log_n T$};

\draw[domain=0:100,smooth,variable=\x,blue] plot ({\x}, {2.373*\x/3 - 2.373*(2/99)*(\x/3)*(\x/3 - 1)/2});
\draw[domain=0:79.35,smooth,variable=\x,blue] plot ({\x}, {2.373*\x/3 - (2/99)*(\x)*(\x - 1)/2});

\node at (98, 0) [below] {$\omega(G)$};
\end{tikzpicture}
\begin{tikzpicture}[scale=0.065]
\tikzstyle{every node}=[font=\footnotesize]
\def\xmin{0}
\def\xmax{108}
\def\ymin{0}
\def\ymax{60}

\fill [green!20, domain=0:61.981, variable=\x]
(0, 58) -- (0, 0)
-- plot ({\x}, {\x - (\x/97)*((\x - 1)/96)*((\x - 2)/95)*(\x - 3)})
-- (98, 46.866)
-- (98, 58)
-- (0, 58);

\fill [blue!20, domain=0:98, variable=\x]
(0, 0)
-- plot ({\x}, {\x - (\x/97)*((\x - 1)/96)*((\x - 2)/95)*(\x - 3)})
-- (0, 0);

\fill [gray!20, domain=61.981:98, variable=\x]
plot ({\x}, {\x - (\x/97)*((\x - 1)/96)*((\x - 2)/95)*(\x - 3)})
-- (98, 46.866)
-- (61.981, 46.866);

\draw[->] (\xmin,\ymin) -- (\xmax,\ymin) node[right] {$k$};
\draw[->] (\xmin,\ymin) -- (\xmin,\ymax) node[above] {$\log_n T$};
\draw[dashed] (61.981, 0) -- (61.981, 46.866);
\node at (61.981, 0) [below] {$k$-clique percolation};
\node at (96, 0) [below] {$\omega(G)$};
\node at (26, 60) [below] {Hypergraphs ($s \ge 3$)};
\node at (12, 40) {feasible};
\node at (35, 5) {infeasible};
\node at (90, 40) {open};
\node [rotate=45] at (20, 25) {$k - \alpha \binom{k}{s}$};
\node at (80, 51) {$\tau + 1 - \alpha \binom{\tau + 1}{s}$};
\draw[blue] (61.981, 46.866) -- (98, 46.866);

\draw[domain=0:98,smooth,variable=\x,blue] plot ({\x}, {\x - (\x/97)*((\x - 1)/96)*((\x - 2)/95)*(\x - 3)});
\end{tikzpicture}

\caption{Comparison of our algorithms and average-case lower bounds for counting $k$-cliques in sparse Erd\H{o}s-R\'{e}nyi Hypergraphs $G(n, c, s)$ with $c = \Theta(n^{-\alpha})$. Green denotes runtimes $T$ feasible for each $k$, blue denotes $T$ infeasible given that the best known worst-case algorithms are optimal and gray denotes $T$ for which the complexity of counting $k$-cliques is open after this work. The left plot shows the graph case of $s = 2$ and the right plot shows the hypergraph case of $s \ge 3$. For simplicity, all quantities shown are up to constant $O_{k, \alpha}(1)$ additive error.}
\label{fig:upperlowerbds}
\end{figure*}

\paragraph{Comparing our upper and lower bounds} A comparison of our algorithmic guarantees and average-case lower bounds based on the best known worst-case algorithms for counting $k$-cliques is shown in Figure \ref{fig:upperlowerbds}.
\begin{enumerate}
\item \textit{Graph Case $(s = 2)$.} In the graph case, our lower and upper bounds have the same form and show that the exponent in the optimal running time is $\frac{\omega k}{3} - C \alpha \binom{k}{2} + O_{k, \alpha}(1)$ where $\frac{\omega}{9} \le C \le 1$ as long as $k \le \kappa + 1 = 2\alpha^{-1} + 1$. As shown in Figure \ref{fig:upperlowerbds}, our upper and lower bounds approach each other for $k$ small relative to $\kappa + 1$.
\item \textit{Hypergraph Case $(s \ge 3)$.} In the hypergraph case of $s \ge 3$, the exponents in our lower and upper bounds are nearly identical at $k - \alpha \binom{k}{s} + O_{k, \alpha}(1)$ up to the $k$-clique percolation threshold.  After this threshold, our lower bounds slowly deteriorate relative to our algorithms until they become trivial at the clique number of $G$ by $k = \kappa + 1$.
\end{enumerate}
Because we consider sparse Erd\H{o}s-R\'{e}nyi hypergraphs, for each $n, k$, and $s$ we actually have an entire family of problems parametrized by the edge probability $c$ and the behavior changes as a function of $c$; this is the first worst-to-average-case hardness result we are aware of for which the complexity of the same problem over worst-case versus average-case inputs is completely different and can be sharply characterized over the whole range of $c$ starting from the same assumption. 
It is surprising that our worst-case to average-case reduction techniques -- which range from the self-reducibility of polynomials to random binary expansions -- together yield tight lower bounds matching our algorithms in the hypergraph case.

Two interesting problems left open by our work are to show average-case lower bounds with an improved constant $C$ in the graph case and to show tight average-case lower bounds beyond the $k$-clique percolation threshold in the case $s \ge 3$. These, other open problems and some extensions of our methods are discussed in Section \ref{sec:openproblems}.

\subsection{Overview of Reduction Techniques}\label{subsec:wcacoverview}

For clarity of exposition, in this section we will restrict our discussion to the graph case $s = 2$, as well as the case of constant $k$. A key step of our worst-case to average-case reduction uses the random self-reducibility of multivariate low-degree polynomials -- i.e., evaluating a polynomial on any worst-case input can be efficiently reduced to evaluating it on several random inputs. This result follows from a line of work \cite{lipton1989new,feigenbaum1993random,gemmell1991self,gemmell1992highly} that provides a method to efficiently compute a polynomial $P : \FF^N \to \FF$ of degree $d \leq |\FF|/20$ on any worst-case input $x \in \FF^N$, given an oracle $\tilde{P} : \FF^N \to \FF$ that agrees with $P$ on a $\frac{1}{2} + \frac{1}{\poly(N)}$ fraction of inputs. Thus, for any low-degree polynomial over a large enough finite field, evaluating the polynomial on a random element in the finite field is roughly as hard as evaluating the polynomial on any adversarially chosen input.

\paragraph{Random self-reducibility for counting $k$-cliques} With the random self-reducibility of polynomials in mind, a natural approach is to express the number of $k$-cliques in a graph as a low-degree polynomial of the $n \times n$ adjacency matrix $A$
$$P(A) = \sum_{\substack{S \subset [n] \\ |S| = k}} \Big(\prod_{i < j \in S} A_{ij}\Big).$$
This polynomial has been used in a number of papers, including by Goldreich and Rothblum \cite{goldreich2018counting} to construct a distribution on dense graphs for which counting $k$-cliques is provably hard on average. However, their techniques are primarily focused on the error probability requirement for the average-case blackbox. As a result, the distribution they obtain is far from Erd\H{o}s-R\'enyi and their approach does not yield tight bounds for sparse graphs. 

The significant obstacle that arises in applying the random self-reducibility of $P$ is that one needs to work over a large enough finite field $\FF_p$, so evaluating $P$ on worst-case graph inputs in $\{0,1\}^{\binom{n}{2}}$ 
only reduces to evaluating $P$ on uniformly random inputs in $\FF_p^{\binom{n}{2}}$.
In order to further reduce to evaluating $P$ on graphs, given a random input $A \in \FF_p^{\binom{n}{2}}$  \cite{goldreich2018counting} uses several gadgets (including replacing vertices by independent sets and taking disjoint unions of graphs) in order to create a larger unweighted random graph $A'$ whose $k$-clique count is equal to $k! \cdot P(A)\pmod{p}$ for appropriate $p$. However, any nontrivial gadget-based reduction seems to have little hope of arriving at something close to the Erd\H{o}s-R\'enyi distribution, because gadgets inherently create non-uniform structure.

\paragraph{Reducing to $k$-partite graphs} We instead consider a different polynomial for graphs on $nk$ vertices with $nk \times nk$ adjacency matrix $A$,
$$P'(A) = \sum_{v_1 \in [n]} \sum_{v_2 \in [2n] \sm [n]} \dots \sum_{v_k \in [kn] \sm [(k-1)n]} \left(\prod_{1 \leq i < j \leq k} A_{v_i v_j}\right).$$ The polynomial $P'$ correctly counts the number of $k$-cliques if $A$ is $k$-partite with vertex $k$-partition $[n] \sqcup ([2n] \sm [n]) \sqcup \dots \sqcup ([kn] \sm [(k-1)n])$. We first reduce clique-counting in the worst case to computing $P'$ in the worst case; this is a simple step, because it is a purely worst-case reduction. 
Next, we construct a recursive counting procedure that reduces evaluating $P'$ on Erd\H{o}s-R\'enyi graphs to counting $k$-cliques in Erd\H{o}s-R\'enyi graphs. Therefore, it suffices to prove that if evaluating $P'$ is hard in the worst case, then evaluating $P'$ on Erd\H{o}s-R\'enyi graphs is also hard.

Applying the Chinese remainder theorem as well as the random self-reducibility of polynomials, computing $P'$ on worst-case inputs in $\{0,1\}^{\binom{nk}{2}}$ reduces to computing $P'$ on several uniformly random inputs in $\FF_p^{\binom{nk}{2}}$, for several different primes $p$ each on the order of $\Theta(\log n)$.  The main question is: how can one evaluate $P'$ on inputs $X \sim \mathrm{Unif}[\FF_p^{\binom{nk}{2}}]$ using an algorithm that evaluates $P'$ on $G(n,c,2)$ Erd\H{o}s-R\'enyi graphs (i.e., inputs $Z \sim \Ber(c)^{\otimes \binom{nk}{2}}$)? 

\paragraph{Eliminating weights with random sparse binary expansions} We solve this by decomposing the random weighted graph $X \sim \mathrm{Unif}[\FF_p^{\binom{nk}{2}}]$ into a weighted sum of graphs $Z^{(0)},\ldots,Z^{(t)} \in \{0,1\}^{\binom{nk}{2}}$ such that each $Z^{(i)}$ is close to Erd\H{o}s-R\'enyi $G(n,c,2)$. Specifically, this additive decomposition satisfies $X \equiv \sum_{i=0}^t 2^i Z^{(i)} \pmod{p}$, i.e., that we can write $X$ as a binary expansion modulo $p$ of Erd\H{o}s-R\'enyi graphs. Importantly, in Section \ref{sec:randombinaryexpansions} we derive near-optimal bounds on $t$ and prove that we can take $t$ to be quite small, growing only as $\poly(c^{-1}(1-c)^{-1} \log(p))$. This technique seems likely to have applications elsewhere. For the unbiased case of $c = 1/2$, a version of this binary expansions technique appeared previously in \cite{goldreich2017worst}.

Now, using the binary expansion decomposition of $X$, we algebraically manipulate $P'$ as follows:
\begin{align*}P'(X) &= \sum_{v_1 \in [n]} \sum_{v_2 \in [2n] \sm [n]} \dots \sum_{v_k \in [kn] \sm [(k-1)n]} \prod_{1 \leq i < j \leq k} \left(\sum_{l \in \{0,\ldots,t\}} 2^l \cdot Z^{(l)}_{v_i v_j}\right) \\ &= \sum_{f \in \{0,\ldots,t\}^{\binom{k}{2}}} \left(\prod_{1 \leq i \leq j \leq k} 2^{f_{ij}}\right) \times \left(\sum_{v_1 \in [n]} \sum_{v_2 \in [2n] \sm [n]} \dots \sum_{v_k \in [kn] \sm [(k-1)n]} \prod_{1 \leq i < j \leq k} Z^{(f_{ij})}_{v_iv_j} \right)
\\&= \sum_{f \in \{0,\ldots,t\}^{\binom{k}{2}}} \left(\prod_{1 \leq i \leq j \leq k} 2^{f_{ij}}\right) P'\left(Z^{(f)}\right).\end{align*}
Here $Z^{(f)}$ is the $nk$-vertex graph with entries given by $Z^{(f_{\bar{a}\bar{b}})}_{ab}$ for $1\leq a< b\leq nk$, where $\bar{a} = \ceil{a/n}$ and $\bar b = \ceil{b/n}$. 
We thus reduce the computation of $P'(X)$ to the computation of a weighted sum of $\poly(c^{-1}(1-c)^{-1} \log(n))^{\binom{k}{2}}$ different evaluations of $P'$ at graphs close in total variation to $G(n,c,2)$. This concludes our reduction.\footnote{If we had instead worked with $P$, then this argument would fail. The argument uses the $k$-partiteness structure of $P'$ as follows: for every pair of vertices $a,b \in [nk]$ and $f \in \{0,\ldots,t\}^{\binom{k}{2}}$, the term $Z_{ab}^{(f_{ij})}$ appearing in the sum is uniquely determined by $a \in [ik] \sm [(i-1)k]$ and $b \in [jk] \sm [(j-1)k]$. So given $f$ we can define a graph $Z^{(f)}$ uniquely. On the other hand, running the same argument with the polynomial $P$, the term $Z_{ab}^{(f_{ij})}$ for many different $i,j$ would appear in the sum, and there is no way to uniquely define a graph $Z^{(f)}$.}

We remark that an important difference between our reduction and the reduction in \cite{goldreich2018counting} is the number of and structure of the calls to the average-case blackbox. Our reduction requires many successful calls to the blackbox in order to obtain a single correct evaluation of the polynomial $P'(A)$, which is where our low error probability requirement comes from. The gadgets in \cite{goldreich2018counting} are specifically designed to only require a single successful call to obtain a single correct evaluation of $P(A)$. Thus even given a blackbox with a constant error probability, the Berkelamp-Welch algorithm can recover $P(A)$ in the case of \cite{goldreich2018counting}.

We also give a different worst-case to average-case reduction for determining the parity of the number of $k$-cliques in Erd\H{o}s-R\'{e}nyi hypergraphs, as discussed in Sections \ref{subsec:reductionthmstatements} and \ref{sec:averagecasereductionproof}.

\subsection{Related Work on Worst-Case to Average-Case Reductions}

The random self-reducibility of low-degree polynomials serves as the basis for several worst-case to average-case reductions found in the literature. One of the first applications of this method was to prove that the permanent is hard to evaluate on random inputs, even with polynomially-small probability of success, unless $\mathsf{P^{\#P}} = \mathsf{BPP}$ \cite{sudan1997decoding,cai1999hardness}. (Under the slightly stronger assumption that $\mathsf{P^{\#P}} \neq \mathsf{AM}$, and with different techniques, \cite{feige1992hardness} proved that computing the permanent on large finite fields is hard even with exponentially small success probability.) Recently, \cite{ball2017average} used the polynomial random self-reducibility result in the fine-grained setting in order to construct polynomials that are hard to evaluate on most inputs, assuming fine-grained hardness conjectures for problems such as \textsc{3-SUM}, \textsc{Orthogonal-Vectors}, and/or \textsc{All-Pairs-Shortest-Paths}. The random self-reducibility of polynomials was also used by Gamarnik and K{\i}z{\i}lda{\u{g}} \cite{gamarnik2018computing} in order to prove that exactly computing the partition function of the Sherrington-Kirkpatrick model in statistical physics is hard on average.

If a problem is random self-reducible, then random instances of the problem are essentially as hard as worst-case instances, and therefore one may generate a hard instance of the problem by simply generating a random instance. Because of this, random self-reducibility plays an important role in cryptography: it allows one to base cryptographic security on random instances of a problem, which can generally be generated efficiently. A prominent example of a random-self reducible problem with applications to cryptography is the problem of finding a short vector in a lattice. In a seminal paper, Ajtai \cite{ajtai1996generating} gave a worst-case to average-case reduction for this short-vector problem. His ideas were subsequently applied to prove the average-case hardness of the Learning with Errors (LWE) problem, which underlies lattice cryptography \cite{ajtai1996generating,regev2009lattices}. A good survey covering worst-case to average-case reductions in lattice cryptography is \cite{regev2010learning}.

There are known restrictions on problems that are self-reducible. For example, non-adaptive worst-case to average-case reductions for $\mathsf{NP}$-complete problems fail unless $\mathsf{coNP} \subseteq \mathsf{NP/poly}$ \cite{feigenbaum1993random, bogdanov2006worst, bogdanov2006average}.

\latestedits{\paragraph{Subsequent work} Several new results have been proved subsequent to the first appearance of our work. Goldreich \cite{goldreich2020counting} provided a simpler reduction for counting the parity of the number of cliques in the uniform $G(n,1/2)$ Erd\H{o}s-R\'enyi graph case. Goldreich obtained error tolerance $\exp(-k^2)$ in this case, which is an improvement over the error tolerance $\exp(-\tilde{O}(k^2))$ in our Theorem~\ref{thm:averagecasehardnessparity}. Hirahara and Shimizu \cite{hirahara2021nearly} studied the average-case complexity of counting bicliques in uniformly random bipartite graphs, obtaining near-optimal runtime bounds assuming the Strong Exponential Time Hypothesis. And Dalirrooyfard, Lincoln, and Vassilevska Williams \cite{dalirrooyfard2020new} extended our techniques to obtain average-case hardness for counting the number of copies of any graph $H$ as an induced subgraph of an Erd\H{o}s-R\'enyi graph $G(n,1/2)$; they also used these techniques to show that simple variations of the orthogonal vectors, $3$-sum and zero-weight $k$-clique problems are hard to count on average for uniform inputs.}

\subsection{Notation and Preliminaries}

A $s$-uniform hypergraph $G = (V(G), E(G))$ consists of a vertex set $V(G)$ and a hyperedge set $E(G) \subseteq \binom{V(G)}{s}$. A $k$-clique $C$ in $G$ is a subset of vertices $C \subset V(G)$ of size $|C| = k$ such that all of the possible hyperedges between the vertices are present in the hypergraph: $\binom{C}{s} \subseteq E(G)$. We write $\cl_k(G)$ to denote the set of $k$-cliques of the hypergraph $G$. One samples from the Erd\H{o}s-R\'enyi distribution $G(n,c,s)$ by independently including each of the $\binom{n}{s}$ hyperedges with probability $c$.

We denote the law of a random variable $X$ by $\cL(X)$. We use $T(A,n)$ to denote the worst-case run-time of an algorithm $A$ on inputs of size parametrized by $n$\latestedits{; for simplicity we assume throughout that $T(A,n)$ is non-decreasing in $n$}.
All algorithms in this paper are randomized, and each (possibly biased) coin flip incurs constant computational cost.

\section{Problem Formulations and Average-Case Lower Bounds}
\label{sec:overview}

\subsection{Clique Problems and Worst-Case Fine-Grained Conjectures} \label{sec:worstcaseproblemslist}
\label{sec:worstcasehardnessconjectures}

In this section, we formally define the problems we consider and the worst-case fine-grained complexity conjectures off of which our average-case lower bounds are based. We focus on the following computational problems.

\begin{definition}
\textsc{\#$(k,s)$-clique} denotes the problem of counting the number of $k$-cliques in an $s$-uniform hypergraph $G$.
\end{definition}

\begin{definition} 
\textsc{Parity-$(k,s)$-clique} denotes the problem of counting the number of $k$-cliques up to parity in an $s$-uniform hypergraph $G$.
\end{definition}

\begin{definition}
\textsc{Decide-$(k,s)$-clique} denotes the problem of deciding whether or not an $s$-uniform hypergraph $G$ contains a $k$-clique.
\end{definition}

Both \textsc{\#$(k,s)$-clique} and \textsc{Decide-$(k,s)$-clique} are fundamental problems that have long been studied in computational complexity theory and are conjectured to be computationally hard in the worst-case setting. When $k$ is allowed to be an unbounded input to the problem, \textsc{Decide-$(k,s)$-clique} is known to be NP-complete \cite{karp1972reducibility} and \textsc{\#$(k,s)$-clique} is known to be \#P-complete \cite{valiant1979complexity}. In this work, we consider the fine-grained complexity of these problems, where $k$ either can be viewed as a constant or a very slow-growing parameter compared to the number $n$ of vertices of the hypergraph. In this context, \textsc{Parity-$(k,s)$-clique} can be interpreted as an intermediate problem between the other two clique problems that we consider. The worst-case reduction from \textsc{Parity-$(k,s)$-clique} to \textsc{\#$(k,s)$-clique} is immediate. As we show in Appendix \ref{sec:decidetoparityreduction}, in the worst-case setting, \textsc{Decide-$(k,s)$-clique} also reduces to \textsc{Parity-$(k,s)$-clique} with a multiplicative overhead of $O(k2^k)$ time.

When $k$ is a constant, the trivial brute-force search algorithms for these problems are efficient in the sense that they take polynomial time. However, these algorithms do not remain efficient under the lens of fine-grained complexity since brute-force search requires $\Theta(n^k)$ time, which can grow significantly as $k$ grows. In the hypergraph case of $s \ge 3$, no algorithm taking time $O(n^{k-\eps})$ on any of these problems is known, including for \textsc{Decide-$(k,s)$-clique} \cite{yuster2006finding}. In the graph case of $s = 2$, the fastest known algorithms for all of these problems take $\Theta(n^{\omega \ceil{k /3}})$ time, where $2 \leq \omega < 2.4$ is the fast matrix multiplication constant \cite{itai1978finding,nevsetvril1985complexity}. Since this is the state of the art, one may conjecture that \textsc{Decide-$(k,s)$-clique} and \textsc{\#$(k,s)$-clique} take $n^{\Omega(k)}$ time in the worst case. 

Supporting this conjecture, Razborov \cite{razborov1985lower} proves that monotone circuits require $\tilde{\Omega}(n^k)$ operations to solve \textsc{Decide-$(k,2)$-clique} in the case of constant $k$. Monotone circuit lower bounds are also known in the case when $k = k(n)$ grows with $n$ \cite{alon1987monotone,amano2005superpolynomial}. In \cite{downey1995fixed}, \textsc{Decide-$(k,2)$-clique} is shown to be $\mathsf{W[1]}$-hard. In other words, this shows that if \textsc{Decide-$(k,2)$-clique} is fixed-parameter tractable -- admits an algorithm taking time $f(k) \cdot \poly(n)$ -- then any algorithm in the parametrized complexity class $\mathsf{W[1]}$ is also fixed-parameter-tractable. This provides further evidence that \textsc{Decide-$(k,2)$-clique} is intractable for large $k$. Finally, \cite{chen2006strong} shows that solving \textsc{Decide-$(k,2)$-clique} in $n^{o(k)}$ time is ETH-hard for constant $k$\footnote{These hardness results also apply to \textsc{Decide-$(k,s)$-clique} for $s \ge 3$ since there is a reduction from \textsc{Decide-$(k,2)$-clique} to \textsc{Decide-$(k,s)$-clique} in $n^s$ time. The reduction proceeds by starting with a graph $G$ and constructing an $s$-uniform hypergraph $G'$ that contains a $s$-hyperedge for every $s$-clique in $G$. The $k$-cliques of $G$ and $G'$ are in bijection. This construction also reduces \textsc{\#$(k,2)$-clique} to \textsc{\#$(k,s)$-clique}.}. We therefore conjecture that the $k$-clique problems take $n^{\Omega(k)}$ time on worst-case inputs when $k$ is constant, as formalized below.

\begin{conjecture}[Worst-case hardness of \textsc{\#$(k,s)$-clique}] \label{conj:weakworstcasecounting} Let $k$ be constant. Any randomized algorithm $A$ for \textsc{\#$(k,s)$-clique} with error probability less than $1/3$ takes time at least $n^{\Omega(k)}$ in the worst case for hypergraphs on $n$ vertices.
\end{conjecture}

\begin{conjecture}[Worst-case hardness of \textsc{Parity-$(k,s)$-clique}] \label{conj:weakworstcaseparity} Let $k$ be constant. Any randomized algorithm $A$ for \textsc{Parity-$(k,s)$-clique} with error probability less than $1/3$ takes time at least $n^{\Omega(k)}$ in the worst case for hypergraphs on $n$ vertices.
\end{conjecture}

\begin{conjecture}[Worst-case hardness of \textsc{Decide-$(k,s)$-clique}] \label{conj:weakworstcasedeciding} Let $k$ be constant. Any randomized algorithm $A$ for \textsc{Decide-$(k,s)$-clique} with error probability less than $1/3$ takes time at least $n^{\Omega(k)}$ in the worst case for hypergraphs on $n$ vertices.
\end{conjecture}

The conjectures are listed in order of increasing strength. Since Conjecture \ref{conj:weakworstcasedeciding} is implied by rETH, they all follow from rETH. We also formulate a stronger version of the clique-counting hardness conjecture, which asserts that the current best known algorithms for $k$-clique counting are optimal.

\begin{conjecture}[Strong worst-case hardness of \textsc{\#$(k,s)$-clique}]\label{conj:strongworstcasecounting}
Let $k$ be constant. Any randomized algorithm $A$ for \textsc{\#$(k,s)$-clique} with error probability less than $1/3$ takes time $\tilde{\Omega}(n^{\omega \ceil{k/3}})$ in the worst case if $s = 2$ and $\tilde{\Omega}(n^k)$ in the worst case if $s \geq 3$.
\end{conjecture}

\subsection{Average-Case Lower Bounds for Counting $k$-Cliques in $G(n, c, s)$}
\label{subsec:reductionthmstatements} \label{subsec:timeslowdownnecessity}

Our first main result is a worst-case to average-case reduction solving either \textsc{\#$(k,s)$-clique} or \textsc{Parity-$(k,s)$-clique} on worst-case hypergraphs given a blackbox solving the problem on {\em most} Erd\H{o}s-R\'enyi hypergraphs drawn from $G(n, c, s)$. We discuss this error tolerance over sampling Erd\H{o}s-R\'enyi hypergraphs as well as the multiplicative overhead in our reduction below. These results show that solving the $k$-clique problems on Erd\H{o}s-R\'enyi hypergraphs $G(n,c,s)$ is as hard as solving them on worst-case hypergraphs, for certain choices of $k,c$ and $s$. Therefore the worst-case hardness assumptions, Conjectures \ref{conj:weakworstcasecounting}, \ref{conj:weakworstcaseparity} and \ref{conj:strongworstcasecounting}, imply average-case hardness on Erd\H{o}s-R\'enyi hypergraphs for \textsc{\#$(k,s)$-clique} and \textsc{Parity-$(k,s)$-clique}.

\begin{theorem}[Worst-case to average-case reduction for \textsc{\#$(k,s)$-clique}] \label{thm:averagecasehardnesscounting}
There is an absolute constant $C > 0$ such that if we define
$$\Upsilon_{\#}(n,c,s,k) \triangleq \left(C (c^{-1}(1-c)^{-1}) (s\log k + s\log \log n) (\log n)\right)^{\binom{k}{s}}$$
then the following statement holds. Let $A$ be a randomized algorithm for \textsc{\#$(k,s)$-clique} with error probability less than $1/\Upsilon_{\#}$ on hypergraphs drawn from $G(n,c,s)$. Then there exists an algorithm $B$ for \textsc{\#$(k,s)$-clique} that has error probability less than $1/3$ on any hypergraph, such that
$$T(B,n) \leq (\log n) \cdot \Upsilon_{\#} \cdot \left(T(A,nk) + (nk)^s \right),$$ where $T(\cA,\ell)$ denotes the runtime of algorithm $\cA$ on $\ell$-vertex hypergraphs.
\end{theorem}

For \textsc{Parity-$(k,s)$-clique} we also give an alternative reduction with an improved reduction time and error tolerance in the dense case when $c = 1/2$.

\begin{theorem}[Worst-case to average-case reduction for \textsc{Parity-$(k,s)$-clique}] \label{thm:averagecasehardnessparity}
We have that:
\begin{enumerate}
\item There is an absolute constant $C > 0$ such that if we define
$$\Upsilon_{P,1}(n,c,s,k) \triangleq \left(C (c^{-1}(1-c)^{-1}) (s\log k)\left(s\log n + \binom{k}{s} \log \log \binom{k}{s}\right)\right)^{\binom{k}{s}}$$
then the following statement holds. Let $A$ be a randomized algorithm for \textsc{Parity-$(k,s)$-clique} with error probability less than $1/\Upsilon_{P,1}$ on hypergraphs drawn from $G(n,c,s)$. Then there exists an algorithm $B$ for \textsc{Parity-$(k,s)$-clique} that has error probability less than $1/3$ on any hypergraph, such that
$$T(B,n) \leq \Upsilon_{P,1} \cdot \left(T(A,nk) + (nk)^s\right)$$
\item There is an absolute constant $C > 0$ such that if we define
$$\Upsilon_{P,2}(s,k) \triangleq \left(C s \log k\right)^{\binom{k}{s}}$$
then the following statement holds. Let $A$ be a randomized algorithm for \textsc{Parity-$(k,s)$-clique} with error probability less than $1/\Upsilon_{P,2}$ on hypergraphs drawn from $G(n,1/2,s)$. Then there exists an algorithm $B$ for \textsc{Parity-$(k,s)$-clique} that has error probability less than $1/3$ on any hypergraph, such that
$$T(B,n) \leq \Upsilon_{P,2} \cdot \left(T(A,nk) + (nk)^s\right)$$
\end{enumerate}
\end{theorem}

Our worst-case to average-case reductions yield the following fine-grained average-case lower bounds for $k$-clique counting and parity on Erd\H{o}s-R\'enyi hypergraphs based on Conjectures \ref{conj:weakworstcasecounting} and \ref{conj:strongworstcasecounting}. We separate these lower bounds into the two cases of dense and sparse Erd\H{o}s-R\'enyi hypergraphs. We remark that, for all constants $k$, an error probability of less than $(\log n)^{-\omega(1)}$ suffices in the dense case and error probability less than $n^{-\omega(1)}$ suffices in the sparse case.

\begin{corollary}[Average-case hardness of \textsc{\#$(k,s)$-clique} on dense $G(n,c,s)$]\label{cor:averagecasecountingdense} If $k,c,\eps > 0$ are constant, then we have that
\begin{enumerate}
\item Assuming Conjecture \ref{conj:weakworstcasecounting}, then any algorithm $A$ for \textsc{\#$(k,s)$-clique} that has error probability less than $(\log n)^{-\binom{k}{s} - \eps}$ on Erd\H{o}s-R\'enyi hypergraphs drawn from $G(n,c,s)$ must have runtime at least $T(A, n) \ge n^{\Omega(k)}$.
\item Assuming Conjecture \ref{conj:strongworstcasecounting}, then any algorithm $A$ for \textsc{\#$(k,s)$-clique} that has error probability less than $(\log n)^{-\binom{k}{s} - \eps}$ on Erd\H{o}s-R\'enyi hypergraphs drawn from $G(n,c,s)$ must have runtime at least $T(A, n) \ge \tilde{\Omega}\left(n^{\omega\ceil{k/3}}\right)$ if $s = 2$ and $T(A, n) \ge \tilde{\Omega}(n^{k})$ if $s \geq 3$.
\end{enumerate}
\end{corollary}

\begin{corollary}[Average-case hardness of \textsc{\#$(k,s)$-clique} on sparse $G(n,c,s)$]\label{cor:averagecasecountingsparse} 
Let $k, \alpha, \eps > 0$ be constants and $c = \Theta(n^{-\alpha})$. Assuming Conjecture \ref{conj:strongworstcasecounting}, then any algorithm $A$ for \textsc{\#$(k,s)$-clique} that has error probability less than $n^{-\alpha \binom{k}{s} - \eps}$ on Erd\H{o}s-R\'enyi hypergraphs drawn from $G(n,c,s)$ must have runtime at least $T(A, n) \ge \tilde{\Omega}\left(n^{\omega\ceil{k/3} - \alpha\binom{k}{s}}\right)$ if $s = 2$ and $T(A, n) \ge \tilde{\Omega}\left(n^{k - \alpha\binom{k}{s}}\right)$ if $s \geq 3$.
\end{corollary}

\latestedits{We remark that Conjecture \ref{conj:weakworstcasecounting} implies there is a constant $C > 0$ such that a version of Corollary \ref{cor:averagecasecountingsparse} holds with the weaker conclusion that $T(A, n) \ge n^{\Omega(k)}$ for any $\alpha \le C k/\binom{k}{s}$.} For \textsc{Parity-$(k,s)$-clique}, we consider here the implications of Theorem \ref{thm:averagecasehardnessparity} only for $c = 1/2$, since this is the setting in which we obtain substantially different lower bounds than for \textsc{\#$(k,s)$-clique}. As shown, an error probability of $o(1)$ on $G(n,1/2,s)$ hypergraphs suffices for our reduction to succeed.

\begin{corollary}[Average-case hardness of \textsc{Parity-$(k,s)$-clique} on $G(n,1/2,s)$]\label{cor:averagecaseparityhalf}
Let $k$ be constant. Assuming Conjecture \ref{conj:weakworstcaseparity}, there is a small enough constant $\eps \triangleq \eps(k,s)$ such that if any algorithm $A$ for \textsc{Parity-$(k,s)$-clique} has error less than $\eps$ on $G(n,1/2,s)$ then $A$ must have runtime at least $T(A, n) \ge n^{\Omega(k)}$.
\end{corollary}

We remark on one subtlety of our setup in the sparse case. Especially in our algorithms section, we generally restrict our attention to $c = \Theta(n^{-\alpha})$ satisfying $\alpha \le k\binom{k}{s}^{-1} = s\binom{k}{s - 1}^{-1}$, which is necessary for the expected number of $k$-cliques in $G(n, c, s)$ to not tend to zero. However, even when this expectation is decaying, the problem \textsc{\#$(k,s)$-clique} as we formulate it is still nontrivial. The simple algorithm that always outputs zero fails with a polynomially small probability that does not appear to meet the $1/\Upsilon_{\#}$ requirement in our worst-case to average-case reduction. A simple analysis of this error probability can be found in Lemma \ref{lem:cliquenumber}. Note that even when $\alpha > s\binom{k}{s - 1}^{-1}$, $\textsc{greedy-random-sampling}$ and its derivative algorithms in Section \ref{sec:algs} still have guarantees and succeed with probability $1 - n^{-\omega(1)}$. We now discuss the multiplicative overhead and error tolerance in our worst-case to average-case reduction for \textsc{\#$(k,s)$-clique}.

\paragraph{Discussion of the Multiplicative Slowdown $\Upsilon_{\#}$} In the sparse case of $c = \Theta(n^{-\alpha})$, our algorithmic upper bounds in Section \ref{sec:algs} imply lower bounds on the multiplicative overhead factor $\Upsilon_{\#}$ in Theorem \ref{thm:averagecasehardnesscounting}. In the hypergraph case of $s \ge 3$ and below the $k$-clique percolation threshold, it must follow that the overhead is at least $\Upsilon_{\#} = \tilde{\Omega}\left( n^{\alpha \binom{k}{s}} \right) = \tilde{\Omega}\left(c^{-\binom{k}{s}} \right)$. Otherwise, our algorithms combined with our worst-case to average-case reduction would contradict Conjecture \ref{conj:strongworstcasecounting}. Up to $\polylog(n)$ factors, this exactly matches the $\Upsilon_{\#}$ from our reduction. In the graph case of $s = 2$, it similarly must follow that the overhead is at least $\Upsilon_{\#} = \tilde{\Omega}\left( n^{\frac{\omega \alpha}{9} \binom{k}{s}} \right) = \tilde{\Omega}\left(c^{-\frac{\omega}{9}\binom{k}{s}} \right)$ to not contradict Conjecture \ref{conj:strongworstcasecounting}. This matches the $\Upsilon_{\#}$ from our reduction up to a constant factor in the exponent.

\paragraph{Discussion of the Error Tolerance $1/\Upsilon_{\#}$} Notice that our worst-case to average-case reductions in Theorems \ref{thm:averagecasehardnesscounting} and \ref{thm:averagecasehardnessparity} require that the error of the average-case blackbox on Erd\H{o}s-R\'enyi hypergraphs go to zero as $k$ goes to infinity. This error tolerance requirement is unavoidable. When $k = \omega(\log n)$ in the dense Erd\H{o}s-R\'enyi graph case of $G(n, 1/2)$, there is a $k$-clique with at most $\binom{n}{k} 2^{-\binom{k}{2}} = o(1)$ probability by a union bound on $k$-subsets of vertices. So in this regime clique-counting on $G(n,1/2)$ with constant error probability is not hard: the algorithm that always outputs zero achieves $o(1)$ average-case error.

If $k \triangleq 3 \log_2 n$, then the probability of a $k$-clique on $G(n,1/2)$ is less than $\binom{n}{k} 2^{-\binom{k}{2}} \leq 2^{-k^2/6}$. So average-case $k$-clique counting is not hard with error more than $2^{-k^2/6}$. On the other hand, our \textsc{\#$(k,2)$-clique} reduction works with average-case error less than $1/\Upsilon_{\#} = 2^{-\Omega(k^2 \log \log n)}$. And our \textsc{Parity-$(k,2)$-clique} reduction is more lenient, requiring error only less than $2^{-\Omega(k^2 \log \log \log n)}$. Thus, the error bounds required by our reductions are quite close to the $2^{-k^2/6}$ error bound that is absolutely necessary for any reduction in this regime.

In the regime where $k = O(1)$ is constant and on $G(n, 1/2)$, our \textsc{Parity-$(k,2)$-clique} reduction only requires a small constant probability of error and our \textsc{\#$(k,2)$-clique} reduction requires less than a $1/\polylog(n)$ probability of error. We leave it as an intriguing open problem whether the error tolerance of our reductions can be improved in this regime.

Finally, we remark that the error tolerance of the reduction must depend on $c$. The probability that a $G(n,c)$ graph contains a $k$-clique is less than $(nc^{(k-1)/2})^{k}$. For example, if $c = 1/n$ then the probability that there exists a $k$-clique is less than $n^{-\Omega(k^2)}$. As a result, no worst-case to average-case reduction can tolerate average-case error more than $n^{-O(k^2)}$ on $G(n, 1/n)$ graphs. And therefore our reductions for \textsc{\#$(k,2)$-clique} and for \textsc{Parity-$(k,2)$-clique} are close to optimal when $c = 1/n$, because our error tolerance scales as $n^{-O(k^2 \log \log n)}$. 

\section{Worst-Case to Average-Case Reduction for $G(n, c, s)$}
\label{sec:averagecasereductionproof}

In this section, we give our main worst-case to average-case reduction that transforms a blackbox solving \textsc{\#$(k,s)$-clique} on $G(n, c, s)$ into a blackbox solving \textsc{\#$(k,s)$-clique} on a worst-case input hypergraph. This also yields a worst-case to average-case reduction for \textsc{Parity-$(k,s)$-clique} and proves Theorems \ref{thm:averagecasehardnesscounting} and \ref{thm:averagecasehardnessparity}. The reduction involves the following five main steps, the details of which are in Sections \ref{sec:kpartite} to \ref{sec:erdosrenyi}.
\begin{enumerate}
\item Reduce \textsc{\#$(k,s)$-clique} and \textsc{Parity-$(k,s)$-clique} on  general worst-case hypergraphs to the worst-case problems with inputs that are $k$-partite hypergraphs with $k$ parts of equal size.
\item Reduce the worst-case problem on $k$-partite hypergraphs to the problem of computing a low-degree polynomial $P_{n,k,s}$ on $N \triangleq N(n,k,s)$ variables over a small finite field $\FF$.
\item Reduce the problem of computing $P_{n,k,s}$ on worst-case inputs to computing $P_{n,k,s}$ on random inputs in $\FF^N$.
\item Reduce the problem of computing $P_{n,k,s}$ on random inputs in $\FF^N$ to computing $P_{n,k,s}$ on random inputs in $\{0,1\}^N$. This corresponds to \textsc{\#$(k,s)$-clique} and \textsc{Parity-$(k,s)$-clique} on $k$-partite Erd\H{o}s-R\'enyi hypergraphs.
\item Reduce the resulting average-case variants of \textsc{\#$(k,s)$-clique} and \textsc{Parity-$(k,s)$-clique} on $k$-partite Erd\H{o}s-R\'enyi hypergraphs to non-$k$-partite Erd\H{o}s-R\'enyi hypergraphs.
\end{enumerate}
These steps are combined in Section \ref{sec:proofsofmain} to complete the proofs of Theorems \ref{thm:averagecasehardnesscounting} and \ref{thm:averagecasehardnessparity}. Before proceeding to our worst-case to average-case reduction, we establish some definitions and notation, and also give pseudocode for the counting reduction in Figure \ref{fig:pseudocodecounting} -- the parity reduction is similar.

The intermediate steps of our reduction crucially make use of $k$-partite hypergraphs with $k$ parts of equal size, defined below.

\begin{definition}[$k$-Partite Hypergraphs] \label{def:kpartiteness}
Given a $s$-uniform hypergraph $G$ on $nk$ vertices with vertex set $V(G) = [n] \times [k]$, define the vertex labelling
$$L : (i,j) \in [n] \times [k] \mapsto j \in [k]$$
If for all $e = \{u_1,\ldots,u_s\} \in E(G)$, the labels $L(u_1), L(u_2), \dots, L(u_s)$ are distinct, then we say that $G$ is $k$-partite with $k$ parts of equal size $n$.
\end{definition}

In our reduction, it suffices to consider only $k$-partite hypergraphs with $k$ parts of equal size. For ease of notation, our $k$-partite hypergraphs will always have $nk$ vertices and vertex set $[n] \times [k]$. In particular, the edge set of a $k$-partite $s$-uniform hypergraph is an arbitrary subset of
$$E(G) \subseteq \left\{\{u_1,\ldots,u_s\} \subset V(G) : L(u_1),\ldots,L(u_s) \text{ are distinct} \right\}$$
Taking edge indicators yields that the $k$-partite hypergraphs on $nk$ vertices we consider are in bijection with $\{0,1\}^N$, where $N \triangleq N(n,k,s) = \binom{k}{s} n^s$ is the size of this set of permitted hyperedges. Thus we will refer to elements $x \in \{0,1\}^N$ and $k$-partite $s$-uniform hypergraphs on $nk$ vertices interchangeably. This definition also extends to Erd\H{o}s-R\'enyi hypergraphs.

\begin{definition}[$k$-Partite Erd\H{o}s-R\'enyi Hypergraphs] \label{def:erdosrenyikpartite}
The $k$-partite $s$-uniform Erd\H{o}s-R\'enyi hypergraph $G(nk,c,s,k)$ is a distribution over hypergraphs on $nk$ vertices with vertex set $V(G) = [n] \times [k]$. A sample from $G(nk,c,s,k)$ is obtained by independently including hyperedge each $e = \{u_1,\ldots,u_s\} \in E(G)$ with probability $c$ for all $e$ with $L(u_1), L(u_2), \dots, L(u_s)$ distinct.
\end{definition}

Viewing the hypergraphs as elements of $G(nk, c, s, k)$ as a distribution on $\{0,1\}^N$, it follows that $G(nk,c,s,k)$ corresponds to the product distribution $\Ber(c)^{\otimes N}$.

\begin{figure}[t!]
\begin{algbox}
\textbf{Algorithm} \textsc{To-ER-\#}$(G,k,A,c)$

\vspace{2mm}

\textit{Inputs}: $s$-uniform hypergraph $G$ with vertex set $[n]$, parameters $k$, $c$, algorithm $A$ for \textsc{\#$(k,s)$-clique} on Erd\H{o}s-R\'enyi hypergraphs with density $c$.
\begin{enumerate}
\item Construct an $s$-uniform hypergraph $G'$ on vertex set $[n] \times [k]$ by defining 
\begin{align*}
E(G') &= \Big\{\{(v_1,t_1),(v_2,t_2),\dots,(v_s,t_s)\} : \{v_1,\ldots,v_s\} \in E(G) \text{ and } \substack{1 \leq v_1 < v_2 < \cdots < v_s \leq n
\\ \\ 1 \leq t_1 < t_2 < \cdots < t_s \leq k} \Big\}.
\end{align*}
Since $G'$ is $k$-partite, view it  as an indicator vector of edges $G' \in \{0,1\}^N$ for $N \triangleq N(n,k,s) = \binom{k}{s} n^s$.

\item Find the first $T$ primes $12\binom{k}{s} < p_1 < \dots < p_T$ such that $\prod_{i=1}^T p_i > n^k$. 

\item Define $L : (a,b) \in [n] \times [k] \mapsto b \in [k]$, and $$P_{n,k,s}(x) = \sum_{\substack{\{u_1,\ldots,u_k\} \in V(G') \\ L(u_i) = i \ \forall i}} \prod_{\substack{S \subseteq [k] \\ |S| = s}} x_{u_S}$$

For each $1 \leq t \leq T$, compute $P_{n,k,s}(G') \pmod{p_t}$, as follows:

\begin{enumerate}
\item[(1)] Use the procedure of \cite{gemmell1992highly} in order to reduce the computation of $P_{n,k,s}(G') \pmod{p_t}$ to the computation of $P_{n,k,s}$ on $M = 12 \binom{k}{s}$ distinct inputs $x_1,\ldots,x_M \sim \mathrm{Unif}[\FF_{p_t}^N]$.
\item[(2)] For each $1 \leq m \leq M$, compute $P_{n,k,s}(x_m) \pmod{p_t}$ as follows:
\begin{enumerate}
    \item[(i)] Use the rejection sampling procedure of Lemma \ref{lem:samplefrommodp} in order to sample $(\latestedits{\tilde{Z}}^{(0)},\ldots,\latestedits{\tilde{Z}}^{(B)})$ close to $(\Ber(c)^{\otimes N})^{\otimes B}$ in total variation distance, such that $x_m \equiv \sum_{b=0}^{B} 2^b \cdot \latestedits{\tilde{Z}}^{(b)} \pmod{p_t}$. It suffices to take $B = \Theta(c^{-1}(1-c)^{-1} s (\log n)(\log p_t))$.
    \item[(ii)] For each function $a : \binom{[k]}{s} \to \{0,\ldots,B\}$, define $\latestedits{\tilde{Z}}^{(a \circ L)}_S = \latestedits{\tilde{Z}}^{a(L(S))}$ for all $S \in [N] \subset \binom{[n]}{s}$. Note that for each $a$, the corresponding $\latestedits{\tilde{Z}}^{(a \circ L)}$ is approximately distributed as $\Ber(c)^{\otimes N}$. Use algorithm $A$ and the recursive counting procedure of Lemma \ref{lem:erkpartitetoergeneralreductioncounting} in order to compute $P_{n,k,s}(\latestedits{\tilde{Z}}^{(a \circ L)})$ for each $a$.
    \item[(iii)] Set $P_{n,k,s}(G') \leftarrow \sum_{a : \binom{[k]}{s} \to \{0,\ldots,B\}} 2^{|a|_1} \cdot P_{n,k,s}(\latestedits{\tilde{Z}}^{(a \circ L)})$.
\end{enumerate}
\end{enumerate}
\item Since $0 \leq P_{n,k,s}(G') \leq n^k$, use Chinese remaindering and the computations of $P_{n,k,s}(G') \pmod{p_i}$ in order to 
calculate and output $P_{n,k,s}(G')$.
\end{enumerate}
\vspace{1mm}
\end{algbox}
\caption{Reduction $\textsc{To-ER-\#}$ for showing computational lower bounds for average-case \textsc{\#$(k,s)$-clique} on Erd\H{o}s-R\'enyi $G(n,c,s)$ hypergraphs based on the worst-case hardness of \textsc{\#$(k,s)$-clique}.}
\label{fig:pseudocodecounting}
\end{figure}

\subsection{Worst-Case Reduction to $k$-Partite Hypergraphs}
\label{sec:kpartite}

In the following lemma, we prove that the worst-case complexity of \textsc{\#$(k,s)$-clique} and \textsc{Parity-$(k,s)$-clique} are nearly unaffected when we restrict the inputs to be worst-case $k$-partite hypergraphs. This step is important, because the special structure of $k$-partite hypergraphs  will simplify future steps in our reduction.

\begin{lemma}\label{lem:worstcasehardnesskpartite}
Let $A$ be an algorithm for \textsc{\#$(k,s)$-clique}, such that $A$ has error probability less than $1/3$ for any $k$-partite hypergraph $G$ on $nk$ vertices. Then, there is an algorithm $B$ for \textsc{\#$(k,s)$-clique} with error probability less than $1/3$ on any hypergraph $G$ satisfying a runtime upper-bound $T(B,n) \leq T(A,n) + O(k^s n^s)$. Furthermore, the same result holds for \textsc{Parity-$(k,s)$-clique} in place of \textsc{\#$(k,s)$-clique}.
\end{lemma}
\begin{proof}
Let $G$ be an $s$-uniform hypergraph on $n$ vertices. Construct the $s$-uniform hypergraph $G'$ on the vertex set $V(G') = [n] \times [k]$ with edge set
$$E(G') = \left\{\{(v_1,t_1),(v_2,t_2),\dots,(v_s,t_s)\} : \{v_1,\ldots,v_s\} \in E(G) \text{ and } \substack{1 \leq v_1 < v_2 < \cdots < v_s \leq n
\\ \\ 1 \leq t_1 < t_2 < \cdots < t_s \leq k}\right\}$$
The hypergraph $G'$ can be constructed in $O(k^s n^s)$ time. Note that $G'$ is $k$-partite with the vertex partition $L : (i,j) \in [n] \times [k] \mapsto j \in [k]$. There is also a bijective correspondence between $k$-cliques in $G'$ and $k$-cliques in $G$ given by
$$\{v_1,v_2,\ldots,v_k\} \mapsto \{(v_1,1),(v_2, 2),\ldots,(v_k,k)\}$$
where $v_1 < v_2 < \dots < v_k$. Thus, the $k$-partite $s$-uniform hypergraph $G'$ on $nk$ vertices has exactly the same number of $k$-cliques as $G$. It suffices to run $A$ on $G'$ and to return its output.
\end{proof}

A corollary to Lemma \ref{lem:worstcasehardnesskpartite} is that any worst-case hardness for \textsc{\#$(k,s)$-clique} and \textsc{Parity-$(k,s)$-clique} on general $s$-uniform hypergraphs immediately transfers to the $k$-partite case. For instance, the lower bounds of Conjectures \ref{conj:weakworstcasecounting}, \ref{conj:weakworstcaseparity}, and \ref{conj:strongworstcasecounting} imply corresponding lower bounds in the $k$-partite case. Going forward in our worst-case to average-case reduction, we may restrict our attention to $k$-partite hypergraphs without loss of generality.

\subsection{Counting $k$-Cliques as a Low-Degree Polynomial}

A key step in our worst-case to average-case reduction is to express the number of $k$-cliques as a low-degree polynomial in the adjacency matrix. As mentioned in the introduction, a similar step -- but without the $k$-partiteness constraint -- appears in the worst-case to average-case reduction of Goldreich and Rothblum \cite{goldreich2018counting}.

Let $\mathcal{E} \subset \binom{V(G)}{s}$ be the set of possible hyperedges that respect the $k$-partition: i.e., $\mathcal{E} = \{A \in \binom{V(G)}{s} : |L(A)| = s\}$. Let $N \triangleq N(n,k,s) = |\mathcal{E}|$ and identify $\mathcal{E}$ with $[N]$ through a bijection $\pi : [N] \to \mathcal{E}$. To simplify the notation, we will omit the map $\pi$ in the proof, and simply treat $[N]$ and $\mathcal{E}$ as the same set. Thus, each $x \in \{0, 1\}^N$ corresponds to a $k$-partite hypergraph where $x_A$ is the indicator that $A \in \mathcal{E}$ is an edge in the hypergraph. The number of $k$-cliques of a $k$-partite hypergraph $x \in \{0,1\}^N$ is a degree-$D$ polynomial $P_{n,k,s} : \{0,1\}^N \to \ZZ$ where $D \triangleq D(k,s) = \binom{k}{s}$:

\begin{equation}\label{eq:cliquecountpolynomialkpartite}
P_{n,k,s}(x) = \sum_{\substack{\{u_1,\ldots,u_k\} \subset V(G) \\ \forall i \ L(u_i) = i}} \prod_{\substack{S \subset [k] \\ |S| = s}} x_{u_S}
\end{equation}
For any finite field $\FF$, this equation defines $P_{n,k,s}$ as a polynomial over that finite field. For clarity, we write this polynomial over $\FF$ as $P_{n,k,s,\FF} : \FF^N \to \FF$. Observe that for any hypergraph $x \in \{0,1\}^N$, we have that
$$P_{n,k,s,\FF}(x) = P_{n,k,s}(x) \pmod{\mathrm{char}(\FF)}$$
where $\mathrm{char}(\FF)$ is the characteristic of the finite field. We now reduce computing \textsc{\#$(k,s)$-clique} and \textsc{Parity-$(k,s)$-clique} on a $k$-partite hypergraph $x \in \{0,1\}^N$ to computing $P_{n,k,s,\FF}(x)$ for appropriate finite fields $\FF$. This is formalized in the following two propositions.

\begin{proposition}\label{prop:countingchineseremaindertheorem}
Let $x \in \{0,1\}^N$ denote a $s$-uniform hypergraph that is $k$-partite with vertex labelling $L$. Let $p_1,p_2,\ldots,p_t$ be $t$ distinct primes, such that $\latestedits{\prod_{i} p_i \in (n^k,n^{2k})}$. First, solving \textsc{\#$(k,s)$-clique} reduces to computing $P_{n,k,s,\FF_{p_i}}(x)$ for all $i \in [t]$, plus $O(\latestedits{(k \log n)^2})$ additive computational overhead. Second, computing $P_{n,k,s,\FF_{p_i}}(x)$ for all $i \in [t]$ reduces to computing \textsc{\#$(k,s)$-clique}, plus $O(t k \log n)$ computational overhead.
\end{proposition}

\begin{proof}
\latestedits{For any $i \in [t]$, it holds that $P_{n,k,s,\FF_{p_i}}(x) \equiv P_{n,k,s}(x) \pmod{p_i}$, which proves the second item of the proposition. The first item follows since $P_{n,k,s}(x) \leq n^k$ because there are at most $n^k$ cliques in the hypergraph. Thus, $P_{n,k,s}(x)$ can be reconstructed from $P_{n,k,s}(x) \pmod{p_i}$ for all $i \in [t]$ in time $O((k \log n)^2)$ by the computational version of the Chinese remainder theorem (Theorem 4.6 of \cite{shoup2009computational}).}
\end{proof}
 
\begin{proposition}\label{prop:parityequivchartwoevaluation}
Let $\FF$ be a finite field of characteristic $2$. Let $x \in \{0,1\}^N$ be a $s$-uniform hypergraph that is $k$-partite with vertex labelling $L$. Then solving \textsc{Parity-$(k,s)$-clique} for $x$ is equivalent to computing $P_{n,k,s,\FF}(x)$.
\end{proposition}

\begin{proof}
This is immediate from $P_{n,k,s,\FF}(x) \equiv P_{n,k,s}(x) \pmod{\mathrm{char}(\FF)}$.
\end{proof}

\subsection{Random Self-Reducibility: Reducing to Random Inputs in $\FF^N$}

Expressing the number and parity of cliques as low-degree polynomials allows us to perform a key step in the reduction: because polynomials over finite fields are random self-reducible, we can reduce computing $P_{n,k,s,\FF}$ on worst-case inputs to computing $P_{n,k,s,\FF}$ on several uniformly random inputs in $\FF^N$.

The following well-known lemma states the random self-reducibility of low-degree polynomials. The lemma first appeared in \cite{gemmell1992highly}. We follow the proof of \cite{ball2017average} in order to present the lemma with explicit guarantees on the running time of the reduction.

\begin{lemma}[Theorem 4 of \cite{gemmell1992highly}] \label{lem:reedsolomon} Let $\FF$ be a finite field with $|\FF| = q$ elements.
Let $N > 0$ and $1 \leq D < q/12$. Let $f : \FF^N \to \FF$ be a polynomial of degree at most $D$. If there is an algorithm $A$ running in time $T(A,N)$ such that $$\PP_{x \sim \mathrm{Unif}\left[\FF^N\right]} [A(x) = f(x)] > 2/3,$$ then there is an algorithm $B$ running in time $O((N+D^2)D \log^2 q + T(A,N) \cdot D)$ such that for {\em any} $x \in \FF^N$, it holds that $\PP[B(x) = f(x)] > 2/3$.
\end{lemma}

For completeness, we provide a proof of this lemma in Appendix \ref{app:reedsolomonproof}. Lemma \ref{lem:reedsolomon} implies that if we can efficiently compute $P_{n,k,s,\FF}$ on at least a 2/3 fraction of randomly chosen inputs in $\FF^N$, then we can efficiently compute the polynomial $P_{n,k,s,\FF}$ over a worst-case input in $\FF^N$.

\subsection{Reduction to Evaluating the Polynomial on $G(nk,c,s,k)$}

So far, we have reduced worst-case clique-counting over unweighted hypergraphs to the average-case problem of computing $P_{n,k,s,\FF}$ over $k$-partite hypergraphs with random edge weights in $\FF$. It remains to reduce from computing $P_{n,k,s,\FF}$ on inputs $x \sim \text{Unif}\left[\FF^N\right]$ to random hypergraphs, which correspond to $x \sim \text{Unif}\left[\{0,1\}^N\right]$. Since $\{0, 1\}^N$ is an exponentially small subset of $\FF^N$ if $|\FF| > 2$, the random weighted and unweighted hypergraph problems are very different. In this section, we carry out this reduction using two different arguments for \textsc{Parity-$(k,s)$-clique} and \textsc{\#$(k, s)$-clique}. The latter reduction is based on the total variation convergence of random binary expansions modulo $p$ to $\text{Unif}[\FF_p]$ and related algorithmic corollaries from Section \ref{sec:randombinaryexpansions}.

We first present the reduction that will be applied in the case of \textsc{Parity-$(k,s)$-clique}. Recall $D = \binom{k}{s}$ is the degree of $P_{n,k,s}$. The following lemma will be used only for the \textsc{Parity-$(k,s)$-clique} case:

\begin{lemma}\label{lem:fpttofp}
Let $p$ be prime and $t \ge 1$. Suppose $A$ is an algorithm that computes $P_{n,k,s,\FF_p}(y)$ with error probability less than $\delta \triangleq \delta(n)$ for $y \sim \textnormal{Unif}\left[ \FF_p^N \right]$ in time $T(A, n)$. Then there is an algorithm $B$ that computes $P_{n,k,s,\FF_{p^t}}(x)$ with error probability less than $t^D \cdot \delta$ for $x \sim \textnormal{Unif}\left[ \FF_{p^t}^N \right]$ in time $T(B,n) = O\left(N t^4 (\log p)^3 + t^D \cdot T(A,n)\right)$.
\end{lemma}

\begin{proof}
We give a reduction computing $P_{n,k,s,\FF_{p^t}}(x)$ where $x \sim \textnormal{Unif}\left[ \FF_{p^t}^N \right]$ given blackbox access to $A$. Let $\beta$ be such that $\beta, \beta^p, \beta^{p^2}, \ldots,\beta^{p^{t-1}} \in \FF_{p^t}$ forms a normal basis for $\FF_{p^t}$ over $\FF_p$. Now for each $i \in [N]$, compute the basis expansion
$$x_i = x_i^{(0)} \beta + x_i^{(1)} \beta^p + \dots + x_i^{(t-1)} \beta^{p^{t-1}}.$$

One can find a generator for a normal basis $\beta \in \FF_{p^t}$ in time $O((t^2 + \log p)(t \log p)^2)$ by Bach et al. \cite{bach1993factor}. Computing $x^{(0)},\ldots,x^{(t-1)}$ then takes time $O(N t^3 (\log p)^3)$ because $N$ applications of Gaussian elimination each take at most $O(t^3)$ operations over $\FF_p$. \footnote{For a good survey on normal bases, we recommend \cite{gao1993normal}.} Note that since $x$ is uniformly distributed and $\beta, \beta^p, \ldots,\beta^{p^{t-1}}$ form a basis, it follows that $x^{(0)},x^{(1)},\ldots,x^{(t-1)}$ are distributed i.i.d according to $\text{Unif}\left[\FF_p^N\right]$.

Given a coloring of the hyperedges $b : [N] \to \{0,1,\ldots,t-1\}$, define $x^{(b)} \in \FF_p^N$ as $x_i^{(b)} = x_i^{(b(i))}$ for all $i \in [N]$. Observe that for any fixed coloring $b$, the vector $x^{(b)}$ is uniform in $\FF_p^N$.

In our proof, for every map $a : \binom{[k]}{s} \to \{0, 1, \ldots,t-1\}$, we construct a coloring $a \circ L : [N] \to \{0,\ldots,t-1\}$ of the hyperedges $[N]$ using the $k$-partiteness of the hypergraph. Given a hyperedge $W = \{w_1,\ldots,w_s\} \in \mathcal{E} = [N]$, we have that $L(W) \in \binom{[k]}{s}$ by the $k$-partiteness of the hypergraph, and hence the color $(a \circ L)(W) \triangleq a(L(W))$ is well-defined. As above, for any fixed $a$, the vector $x^{(a \circ L)}$ is uniform in $\FF_p^N$.

We now manipulate $P_{n,k,s,\FF_{p^t}}$. First we write each entry $x_{u_S}$ in the normal basis, and then we redistribute terms to write $P_{n,k,s,\FF_{p^t}}$ as a weighted sum of clique-counts modulo $p$:

\allowdisplaybreaks
\begin{align*}
P_{n,k,s,\FF_{p^t}}(x) &= \sum_{\substack{\{u_1,\ldots,u_k\} \subset V(G) \\ \forall j \ L(u_j) = j}} \prod_{S \in \binom{[k]}{s}} x_{u_S} \\
&= \sum_{\substack{\{u_1,\ldots,u_k\} \subset V(G) \\ \forall j \ L(u_j) = j}} \prod_{S \in \binom{[k]}{s}} \left(\sum_{i=0}^{t-1} x_{u_S}^{(i)} \beta^{p^i}\right) \\
&= \sum_{ a : {[k] \choose s} \to \{0,\ldots,t-1\}} \left(\sum_{\substack{\{u_1,\ldots,u_k\} \subset V(G) \\ \forall i \ L(u_i) = i}} \prod_{S \in \binom{[k]}{s}} \left( x_{u_S}^{(a(S))} \beta^{p^{a(S)}}\right)\right) \\
&= \sum_{ a : {[k] \choose s}  \to \{0,\ldots,t-1\}} \left(\prod_{S \in \binom{[k]}{s}} \beta^{p^{a(S)}}\right) \left(\sum_{\substack{\{u_1,\ldots,u_k\} \subset V(G) \\ \forall i \ L(u_i) = i}} \prod_{S \in \binom{[k]}{s}} x_{u_S}^{(a(S))}\right) \\
&= \sum_{ a : {[k] \choose s}  \to \{0,\ldots,t-1\}} \left(\prod_{S \in \binom{[k]}{s}} \beta^{p^{a(S)}}\right) P_{n,k,s,\FF_p}\left(x^{(a \circ L)}\right)
\end{align*}
Since $x^{(a \circ L)} \sim \text{Unif}\left[ \FF_p^N \right]$ for each fixed map $a$, computing $P_{n,k,s,\FF_{p^t}}(x)$ reduces to evaluating $P_{n,k,s,\FF_p}$ on $t^D$ uniformly random inputs in $\FF_p^N$ and outputting a weighted sum of the evaluations. The error probability is bounded by a union bound.
\end{proof}

We now give the reduction to evaluating $P_{n, k, s}$ on random hypergraphs drawn from $G(nk, c, s, k)$ in the case of \textsc{\#$(k, s)$-clique}. One of the main lemmas driving the reduction is the following:
\begin{lemma} \label{lem:samplefrommodp}
\latestedits{There is an absolute constant $K > 0$ such that the following holds. Let $p > 2$ be prime, $\eps > 0$, $c \in (0,1)$, and $t \ge K \cdot c^{-1} (1-c)^{-1} \log (p/\epsilon) \log p$. Then there is an $O(pt \log(1/\eps) \log(p))$-time algorithm that, given $x \in \FF_p$, samples a random variable $\tilde{Z}_x = (\tilde{Z}_x^{(0)},\ldots,\tilde{Z}_x^{(t-1)}) \in \{0,1\}^{t}$ satisfying $\sum_{i=0}^{t-1} 2^i \cdot \tilde{Z}_{x}^{(i)} \equiv x \pmod{p}\mbox{ almost surely}.$ Moreover, if $x \sim \mathrm{Unif}[\FF_p]$ then $\dtv(\cL(\tilde{Z}_{x}),\Ber(c)^{\otimes t}) \leq \eps.$} 
\end{lemma}

The proof of Lemma~\ref{lem:samplefrommodp} is deferred to Section~\ref{sec:randombinaryexpansions}. It is a central ingredient in the \textsc{\#$(k,s)$-clique} reduction and will be used through the following lemma.
\begin{lemma}\label{lem:fptogncsk}
Let $p$ be prime and let $c = c(n), \gamma = \gamma(n) \in (0,1)$. Suppose that $A$ is an algorithm that computes $P_{n,k,s,\FF_p}(y)$ with error probability less than $\delta \triangleq \delta(n)$ when $y \in \{0,1\}^N$ is drawn from $G(nk,c,s,k)$. Then, for some $t = O(c^{-1}(1-c)^{-1} \log(Np/\gamma) \log p)$, there is an algorithm $B$ that evaluates $P_{n,k,s,\FF_p}(x)$ with error probability at most $\gamma + t^D\cdot\delta$ when $x\sim \textnormal{Unif}\left[\FF_p^N\right]$ in time upper bounded by $T(B,n) = O\left(N p t \log(Np/\gamma)\latestedits{\log(p)} + t^D \cdot T(A,n)\right)$.
\end{lemma}

\begin{proof}
\allowdisplaybreaks
We give a reduction computing $P_{n,k,s,\FF_p}(x)$ where $x \sim \textnormal{Unif}\left[ \FF_{p}^N \right]$ given blackbox access to $A$. We first handle the case in which $p > 2$. For each $j \in [N]$, apply the algorithm from Lemma \ref{lem:samplefrommodp} to sample $\latestedits{\tilde{Z}_j = (\tilde{Z}_j^{(0)}, \tilde{Z}_j^{(1)}, \ldots,\tilde{Z}_j^{(t-1)}) \in \{0,1\}^t}$ satisfying
$$\sum_{i=0}^{t-1} 2^i \latestedits{\cdot \tilde{Z}}_j^{(i)} \equiv x_j \pmod{p} \quad \text{and} \quad \dtv\left(\cL(\latestedits{\tilde{Z}_j}), \Ber(c)^{\otimes t}\right) \le \eps \triangleq \gamma/N$$
By Lemma \ref{lem:samplefrommodp}, we may choose $t = O(c^{-1}(1-c)^{-1} \log(Np/\gamma) \log p)$ and this sampling can be carried out in $O(Npt \log(Np/\gamma)\latestedits{\log(p)})$ time. \latestedits{Now expand $P_{n,k,s,\FF_p}(x)$ in terms of $\tilde{Z}$, similarly to the calculations in Lemma \ref{lem:fpttofp}}. We are working in $\FF_p$ so the following equalities hold modulo $p$:
\begin{align*}P_{n,k,s,\FF_p}(x) &= \sum_{\substack{\{u_1,\ldots,u_k\} \subset V(G) \\ \forall j \ L(u_j) = j}} \prod_{S \in \binom{[k]}{s}} x_{u_S} \\
&= \sum_{\substack{\{u_1,\ldots,u_k\} \subset V(G) \\ \forall j \ L(u_j) = j}} \prod_{S \in \binom{[k]}{s}} \left(\sum_{i=0}^{t-1} 2^i \cdot \latestedits{\tilde{Z}}_{u_S}^{(i)}\right) \\
&= \sum_{ a : {[k] \choose s} \to \{0,\ldots,t-1\}} \left(\sum_{\substack{\{u_1,\ldots,u_k\} \subset V(G) \\ \forall i \ L(u_i) = i}} \prod_{S \in \binom{[k]}{s}} \left( 2^{a(S)} \cdot \latestedits{\tilde{Z}}_{u_S}^{(a(S))}\right)\right) \\
&= \sum_{ a : {[k] \choose s}  \to \{0,\ldots,t-1\}} \left(\prod_{S \in \binom{[k]}{s}} 2^{a(S)}\right) \left(\sum_{\substack{\{u_1,\ldots,u_k\} \subset V(G) \\ \forall i \ L(u_i) = i}} \prod_{S \in \binom{[k]}{s}} \latestedits{\tilde{Z}}_{u_S}^{(a(S))}\right) \\
&= \sum_{ a : {[k] \choose s}  \to \{0,\ldots,t-1\}} \left(\prod_{S \in \binom{[k]}{s}} 2^{a(S)}\right) P_{n,k,s,\FF_p}(\latestedits{\tilde{Z}}^{(a \circ L)}),\end{align*}
where, as in the proof of Lemma \ref{lem:fpttofp}, given any coloring $b : [N] \to \{0,\ldots,t-1\}$, we define $\latestedits{\tilde{Z}}^{(b)} \in \{0,1\}^N$ by $\latestedits{\tilde{Z}}_j^{(b)} = \latestedits{\tilde{Z}}_j^{(b(j))}$, for all $j \in [N]$. \latestedits{Computing $P_{n,k,s,\FF_p}(x)$ thus reduces to computing a weighted sum over the $t^D$ evaluations of $P_{n,k,s,\FF_p}(\tilde{Z}^{(a \circ L)})$ for all maps $a : \binom{[k]}{s} \to \{0,\ldots,t-1\}$. Our algorithm uses the blackbox $A$ to compute each term, and outputs the weighted sum. In other words, our algorithm returns
$$\sum_{ a : {[k] \choose s}  \to \{0,\ldots,t-1\}} \left(\prod_{S \in \binom{[k]}{s}} 2^{a(S)}\right) A(\latestedits{\tilde{Z}}^{(a \circ L)})\,.$$
Let $E$ be the event that the calls to the blackbox are all correct: i.e., $A(\tilde{Z}^{(a \circ L)}) = P_{n,k,s,\FF_p}(\tilde{Z}^{(a \circ L)})\  \mbox{ for all } a : \binom{[k]}{s} \to \{0,\ldots,t-1\}$. If $E$ holds, then our algorithm correctly computes $P_{n,k,s,\FF_p}(x)$. It suffices to prove that $$\PP\left[E\right] > 1 - \gamma + t^D \cdot \delta.$$ For the analysis, note that for each $j \in [N]$, the random vector $(\tilde{Z}_j^{(0)},\ldots,\tilde{Z}_j^{(t-1)})$ may be coupled with $(Z_j^{(0)},\ldots,Z_j^{(t-1)}) \sim \Ber(c)^{\otimes t}$, such that} $$\PP[\latestedits{\tilde{Z}}_j^{(i)} = Z_j^{(i)}\  \forall i,j] \geq 1 - \gamma.$$ Moreover, since $\latestedits{\tilde{Z}}_j^{(i)}$ is independent of $\latestedits{\tilde{Z}}_l^{(k)}$ whenever $j \neq l$, in the coupling we may choose $Z$ such that $Z_j^{(i)}$ is independent of $Z_l^{(k)}$ whenever $j \neq l$. \latestedits{Thus, for any fixed coloring $b : [N] \to \{0,\ldots,t-1\}$,} the entries $Z_1^{(b)},\ldots,Z_N^{(b)}$ are independent and distributed as $\Ber(c)$. In other words, $Z^{(b)} \sim G(nk,c,s,k).$
\latestedits{We use these facts to lower-bound the probability of $E$ as follows: \begin{align*}\PP[E] &\geq \PP[E \mbox{ and } \tilde{Z}_j^{(i)} = Z_j^{(i)}\  \forall i,j] \\
&= \PP[A(\tilde{Z}^{(a \circ L)}) = P_{n,k,s,\FF_p}(\tilde{Z}^{(a \circ L)})\  \forall a, \mbox{ and } \tilde{Z}_j^{(i)} = Z_j^{(i)}\  \forall i,j] \\ 
&= \PP[A(Z^{(a \circ L)}) = P_{n,k,s,\FF_p}(Z^{(a \circ L)})\  \forall a, \mbox{ and } \tilde{Z}_j^{(i)} = Z_j^{(i)}\  \forall i,j] \\
&\geq 1 - (1-\PP[\tilde{Z}_j^{(i)} = Z_j^{(i)}\  \forall i,j]) - \sum_{a : \binom{[k]}{s} \to \{0,\ldots,t-1\}}\PP[A(Z^{(a \circ L)}) \neq P_{n,k,s,\FF_p}(Z^{(a \circ L)})] \\
&> 1 - \gamma - t^D \cdot \delta,\end{align*}
where the second-to-last line is a union bound, and the last line uses that $Z^{(a \circ L)} \sim G(nk,c,s,k)$ for any fixed $a$, and applies the error guarantee of $A$.} This proves correctness of the algorithm for the case $p > 2$.

If $p = 2$, then the proof is almost identical, except that since $2 \equiv 0 \pmod{2}$, we may no longer use the result on random binary expansions of Lemma \ref{lem:samplefrommodp}. In this case, for each $j \in [N]$ we sample $\latestedits{\tilde{Z}_j = (\tilde{Z}}_j^{(0)},\ldots,\latestedits{\tilde{Z}}_j^{(t-1)}\latestedits{) \in \{0,1\}^t}$ such that $$\sum_{i=0}^{t-1} \latestedits{\tilde{Z}}_j^{(i)} \equiv x_j \pmod{p} \quad \mbox{and} \quad \dtv(\cL(\latestedits{\tilde{Z}_j}), \Ber(c)^{\otimes t}) \leq \eps \triangleq \gamma/N.$$ By Lemma \ref{lem:samplefrommodtwo} (deferred but analogous to Lemma~\ref{lem:samplefrommodp}), we may choose $t = O(c^{-1}(1-c)^{-1}\log(N/\gamma))$, and we may sample in time $O(Nt \log(N/\gamma))$. By a similar, and simpler, calculation to the one for the case $p > 2$, we have that $$P_{n,k,s,\FF_2}(x) = \sum_{ a : {[k] \choose s}  \to \{0,\ldots,t-1\}} P_{n,k,s,\FF_2}(\latestedits{\tilde{Z}}^{(a \circ L)}).$$
\latestedits{Our algorithm returns $$\sum_{ a : {[k] \choose s}  \to \{0,\ldots,t-1\}} A(\latestedits{\tilde{Z}}^{(a \circ L)}),$$ which is correct with probability at least $1 - \gamma - t^D \cdot \delta$ similarly to the $p > 2$ case. The proof is again to couple $\tilde{Z}$ with a random variable $Z$ such that $\PP[\tilde{Z}_j^{(i)} = Z_j^{(i)} \ \forall i,j] \geq 1 - \gamma$, and, for each $a$, $Z^{(a \circ L)}$ is distributed as $G(nk,c,s,k)$.}
\end{proof}

\subsection{Reduction to Counting $k$-Cliques in $G(n,c,s)$}
\label{sec:erdosrenyi}

So far, we have reduced \textsc{Parity-$(k,s)$-clique} and \textsc{\#$(k, s)$-clique} for worst-case input hypergraphs to average-case inputs drawn from the $k$-partite Erd\H{o}s-R\'{e}nyi distribution $G(nk,c,s,k)$. We now carry out the final step of the reduction, showing that \textsc{Parity-$(k,s)$-clique} and \textsc{\#$(k, s)$-clique} on inputs drawn from $G(nk,c,s,k)$ reduce to inputs drawn from the non-$k$-partite Erd\H{o}s-R\'{e}nyi distribution $G(n,c,s)$. Recall that a hypergraph $G$ drawn from $G(nk,c,s,k)$ has vertex set $V(G) = [n] \times [k]$ and vertex partition given by the labels $L : (i,j) \in [n] \times [k] \mapsto j \in [k]$.

\begin{lemma}\label{lem:erkpartitetoergeneralreductioncounting}
Let $\delta = \delta(n) \in (0,1)$ be a non-increasing function of $n$ and let $c = c(n) \in (0,1)$. Suppose that $A$ is a randomized algorithm for \textsc{\#$(k,s)$-clique} such that for any $n$, $A$ has error probability less than $\delta(n)$ on hypergraphs drawn from $G(n,c,s)$ in $T(A, n)$ time. Then there exists an algorithm $B$ solving \textsc{\#$(k,s)$-clique} that has error probability less than $2^k \cdot \delta(n)$ on hypergraphs drawn from $G(nk,c,s,k)$ and that runs in $T(B,n) = O\left(2^k \cdot T(A,nk) + k^s n^s + s^2 k^3 2^k \log^2(nk) \right)$ time.
\end{lemma}

\begin{proof}
It suffices to count the number of $k$-cliques in $G \sim G(nk,c,s,k)$ given blackbox access to $A$. Construct the hypergraph $H$ over the same vertex set $V(H) = [n] \times [k]$ by starting with $G$ and adding every edge $e = \{v_1, v_2, \ldots, v_s\} \in \binom{[n] \times [k]}{s}$ satisfying the condition $|\{L(v_1), \ldots, L(v_s)\}| < s$ independently with probability $c$. In other words, independently add each edge to $G$ containing two vertices from the same part of $G$. It follows that $H$ is distributed according to $G(nk,c,s)$. More generally, for every $S \subset [k]$, $H_S$ is distributed according to $G(|L^{-1}(S)|,c,s)$ where $H_S$ is the restriction of $H$ to the vertices $L^{-1}(S) \subset V(H)$ with labels in $S$. Note that $H$ can be constructed in $O(k^s n^s)$ time.

Now observe that for each $S \neq \emptyset$, it holds that $n \le |L^{-1}(S)| \le nk$ and the algorithm $A$ succeeds on each $H_S$ with probability at least $1 - \delta(n)$. By a union bound, we may compute the number of $k$-cliques $|\cl_k(H_S)|$ in $H_S$ for all $S \subset [k]$ with error probability less than $2^k \cdot \delta(n)$. Note that this can be done in $O\left(2^k \cdot T(A,nk)\right)$ time. From these counts $|\cl_k(H_S)|$, we now inductively compute
$$t_d \triangleq |\{S \in \cl_k(H) : |L(S)| = d\}|$$
for each $d \in [k]$. Note that $t_0 = 0$ in the base case $d = 0$. Given $t_0, t_1, \ldots, t_d$, the next count $t_{d+1}$ can be expressed by inclusion-exclusion as
\begin{align*}
t_{d+1} &= \sum_{T \subset [k], |T| = d+1} |\{S \in \cl_k(H) : L(S) = T\}| \\ 
&= \sum_{T \subset [k], |T| = d+1} \left(|\cl_k(H_T)| - \sum_{i=0}^{d} \sum_{U \subset T, |U| = i} |\{S \in \cl_k(H) : L(S) = U\}|\right) \\ 
&= \left(\sum_{T \subset [k], |T| = d+1} |\cl_k(H_T)|\right) - \sum_{i=0}^{d} \binom{k-i}{d+1-i} |\{S \in \cl_k(H) : |L(S)| = i\}| \\ 
&= \sum_{T \subset [k], |T| = d+1} |\cl_k(H_T)| - \sum_{i=0}^d \binom{k-i}{d+1-i} t_i
\end{align*}
After $O(k2^k)$ operations, this recursion yields the number of $k$-cliques $t_k = |\{S \in \cl_k(H) : |L(S)| = k\}| = |\cl_k(G)|$ in the original $k$-partite hypergraph $G$. The sizes of the integers manipulated are always at most $2^k \binom{nk}{s}$, so each arithmetic operation takes $O((ks \log(nk))^2)$ time.
\end{proof}

Repeating the same proof over $\FF_2$ yields an analogue of Lemma \ref{lem:erkpartitetoergeneralreductioncounting} for \textsc{Parity-$(k,s)$-clique}, as stated below.

\begin{lemma}\label{lem:erkpartitetoergeneralreductionparity}
Lemma \ref{lem:erkpartitetoergeneralreductioncounting} holds when \textsc{\#$(k,s)$-clique} is replaced by \textsc{Parity-$(k,s)$-clique}.
\end{lemma}

\subsection{Proofs of Theorems \ref{thm:averagecasehardnesscounting} and \ref{thm:averagecasehardnessparity}}
\label{sec:proofsofmain}

We now combine Steps 1-5 formally in order to prove Theorems \ref{thm:averagecasehardnesscounting} and \ref{thm:averagecasehardnessparity}.

\begin{proof}[Proof of Theorem \ref{thm:averagecasehardnesscounting}]
Our goal is to construct an algorithm $B$ solving \textsc{\#$(k,s)$-clique} with error probability $< 1/3$ on any $s$-uniform hypergraph $x$. We are given an algorithm $A$ that solves \textsc{\#$(k,s)$-clique} with probability of error $< 1/\Upsilon_{\#}$ on hypergraphs drawn from $G(n,c,s)$. We will construct the following intermediate algorithms in our reduction:

\begin{itemize} \item Algorithm $A_0$ that solves \textsc{\#$(k,s)$-clique} with error probability $< 1/3$ for any worst-case $k$-partite hypergraph.
\item Algorithm $A_1(x,p)$ that computes $P_{n,k,s,\FF_p}(x)$ for any $x \in \FF_p^N$ and for any prime $p$ such that $12 \binom{k}{s} < p < 10 \log n^k$, with  worst-case error probability $< 1/3$.
\item Algorithm $A_2(y,p)$ for primes $12 \binom{k}{s} < p < 10 \log n^k$ computing $P_{n,k,s,\FF_p}(y)$ on inputs $y \sim \mathrm{Unif}[\FF_p^N]$ with error probability $< 1/3$.
\item Algorithm $A_3(z)$ that computes $P_{n,k,s}(z)$ on inputs $z \sim G(nk,c,s,k)$ with error probability $< \delta$. (The required value of $\delta$ will be determined later on.)
\end{itemize}

We construct algorithm $B$ from $A_0$, $A_0$ from $A_1$, $A_2$ from $A_3$, and $A_3$ from $A$.
\\

\textit{1. Reduce to computing \textsc{\#$(k,s)$-clique} for $k$-partite hypergraphs.} We use Lemma \ref{lem:worstcasehardnesskpartite} to construct $B$ from $A_0$, such that $B$ runs in time $$T(B,n) = T(A_0,n) + O((nk)^s).$$

\textit{2. Reduce to computing $P_{n,k,s,\FF_p}$ on worst-case inputs.} We use Proposition \ref{prop:countingchineseremaindertheorem} to construct $A_0$ from $A_1$ such that $A_0$ runs in time $$T(A_0,n) \leq O(T(A_1,n) \cdot \log n^k  + (\log n^k)^2).$$ The algorithm $A_0$ starts by using a sieve to find the first $T$ primes $12\binom{k}{s} < p_1 < \dots < p_T$ such that $\prod_{i=1}^T p_i > n^k$. Notice that $p_T \leq 10 \log n^k$, so this step takes time $O((\log n^k)^2)$. Then, given a $k$-partite hypergraph $x \in \{0,1\}^N$, the algorithm $A_0$ computes $P_{n,k,s}(x)$ by first computing $P_{n,k,s,\FF_{p_i}}(x)$ for all $p_i$ with algorithm $A_1$, boosting the error of $A_1$ by repetition and majority vote. Since $T = O((\log n^k) / (\log \log n^k))$, we only need to repeat $O(\log \log n^k)$ times per prime; this yields a total slowdown factor of $O(\log n^k)$. \latestedits{Finally, $P_{n,k,s}(x)$, the number of $k$-cliques in $x$, is computed from the values of $P_{n,k,s,\FF_{p_i}}(x)$ in $O((k \log n)^2)$ time by the computational Chinese remainder theorem stated in Proposition~\ref{prop:countingchineseremaindertheorem}.}

\textit{3. Reduce to computing $P_{n,k,s,\FF_p}$ on random inputs in $\FF_p^N$.} We use Lemma \ref{lem:reedsolomon} to construct $A_1$ from $A_2$ such that $A_1$ runs in time
\begin{align*}
T(A_1,n) &= O((N+D^2)D \log^2p + D \cdot T(A_2,n)) \\
&= O(n^s \binom{k}{s}^3 \log^2 \log n^k + \binom{k}{s} \cdot T(A_2,n)).
\end{align*}

\textit{4. Reduce to computing $P_{n,k,s}$ on random inputs in $\{0,1\}^N$.} We use Lemma \ref{lem:fptogncsk} to construct $A_2$ from $A_3$ such that $A_2$ runs in time $$T(A_2,n) = O(Np\latestedits{t_p} \log (Np)\latestedits{\log(p)} + t_p^{\binom{k}{s}} \cdot T(A_3,n)),$$ for some $\latestedits{t_p} = O(c^{-1}(1-c)^{-1} s (\log n) (\log p)).$ For this step, we require the error probability $\delta$ of algorithm $A_3(z)$ on inputs $z \sim G(nk,c,s,k)$ to be at most $1/(4\latestedits{t_p}^D) = 1 /(4 \latestedits{t_p}^{\binom{k}{s}})$. \latestedits{Recall that we always have $p = O(k \log n)$ in this step, and hence $t_p$ is upper-bounded by a uniform value $t = \Theta(c^{-1}(1-c)^{-1} s (\log n) (\log k + \log \log n))$.}

\textit{5. Reduce to computing \textsc{\#$(k,s)$-clique} for $G(n,c,s)$ hypergraphs.}
We use Lemma \ref{lem:erkpartitetoergeneralreductioncounting} to construct $A_3$ from $A$ such that $A_3$ runs in time $$T(A_3,n) = O((nk)^s + s^2 k^3 2^k\log^2(nk) + 2^k \cdot T(A,nk)),$$ and such that $A_3$ has error probability at most $\delta < 2^k / \Upsilon_{\#}$.
\\

As in the theorem statement, let $\Upsilon_{\#}(n,c,s,k) \triangleq (C (c^{-1}(1-c)^{-1}) s (\log n)(\log k + \log \log n))^{\binom{k}{s}}$, where $C > 0$ is a large constant to be determined. If we take $C$ large enough, then \latestedits{$\Upsilon_{\#} \geq (10t)^{\binom{k}{s}}$}. In this case, \latestedits{since $\binom{k}{s} \geq k \geq 3$ without loss of generality}, the error $\delta$ of $A_3$ will be at most \latestedits{$\delta \leq 2^k / \Upsilon_{\#} \leq 1/(5t)^{\binom{k}{s}} < 1/(4 t^{\binom{k}{s}})$}, which is what we needed for the fourth step. It remains to put the runtime bounds together, 
\begin{align*}
T(B,n) &= O\Big((nk)^s + (\log n^k)^2 + (\log n^k) \cdot \Big(n^s k\latestedits{^2} \binom{k}{s}^3 (\log n)^2 \\
&\quad \quad + \binom{k}{s} \cdot \Big(\latestedits{N(k \log n)t \log (N)\log(k \log n) +  4^k t^{\binom{k}{s}}} \cdot (T(A,nk) + (nk)^s)\Big)\Big)\Big) \\
&= O\Big(n^s k^3 \binom{k}{s}^3  (c^{-1}(1-c)^{-1})s (\log n)^4 (\log k + \log \log n)\latestedits{^2} \\
&\quad \quad + (\log n) \cdot \Upsilon_{\#} \cdot (T(A,nk) + (nk)^s)\Big) \\
&= \latestedits{O((\log n) \cdot (10t)^{\binom{k}{s}} \cdot n^s + (\log n) \cdot \Upsilon_{\#} \cdot (T(A,nk) + (nk)^s))},
\end{align*}
\latestedits{where we have used that $\binom{k}{s} \geq 3$ without loss of generality. The last term dominates since $\Upsilon_{\#} \geq (10t)^{\binom{k}{s}}$}, and thus $$T(B,n)= O((\log n) \cdot \Upsilon_{\#} \cdot (T(A,nk) + (nk)^s)).$$
\end{proof}

\begin{proof}[Proof of Theorem \ref{thm:averagecasehardnessparity}]
The proof of item 1 of Theorem \ref{thm:averagecasehardnessparity} is analogous to the proof of Theorem \ref{thm:averagecasehardnesscounting}, except that it does not use the Chinese remainder theorem \latestedits{(Proposition~\ref{prop:countingchineseremaindertheorem})}. Moreover, special care is needed in order to ensure that the field $\FF$ over which we compute the polynomial $P_{n,k,s,\FF}$ in the intermediate steps is large enough that we may use the random self-reducibility of polynomials.

Our goal is to construct an algorithm $B$ that solves \textsc{Parity-$(k,s)$-clique} with error probability $< 1/3$ on any $s$-uniform hypergraph $x$. We are given an algorithm $A$ that solves \textsc{Parity-$(k,s)$-clique} with probability of error $< 1/\Upsilon_{P,1}$ on hypergraphs drawn from $G(n,c,s)$. We will construct the following intermediate algorithms in our reduction:

\begin{itemize} \item Algorithm $A_0$ that solves \textsc{Parity-$(k,s)$-clique} with error probability $< 1/3$ for any worst-case $k$-partite hypergraph.
\item Algorithm $A_1(w)$ that computes $P_{n,k,s,\FF_{2^{\kappa}}}(w)$ on inputs $w \sim \mathrm{Unif}[\FF_{2^{\kappa}}^N]$ for $\kappa = \ceil{\log_2 (12 \binom{k}{s})}$, with error probability $< 1/3$.
\item Algorithm $A_2(y)$ that computes $P_{n,k,s,\FF_2}(y)$ on inputs $y \sim \mathrm{Unif}[\FF_2^N]$ with error probability $< \delta_2$. (The required value of $\delta_2$ will be determined later on.)
\item Algorithm $A_3(z)$ that computes $P_{n,k,s,\FF_2}(z)$ on inputs $z \sim G(nk,c,s,k)$ with error probability $< \delta_3$. (The required value of $\delta_3$ will be determined later on.)
\end{itemize}

We construct algorithm $B$ from $A_0$, $A_0$ from $A_1$, $A_2$ from $A_3$, and $A_3$ from $A$.
\\

\textit{1. Reduce to computing \textsc{Parity-$(k,s)$-clique} for $k$-partite hypergraphs.} We use Lemma \ref{lem:worstcasehardnesskpartite} to construct $B$ from $A_0$, such that $B$ runs in time $$T(B,n) = T(A_0,n) + O((nk)^s).$$

\textit{2. Reduce to computing $P_{n,k,s,\FF_{2^{\kappa}}}$ on random inputs in $\FF_{2^{\kappa}}^N$.} Note that by Proposition \ref{prop:parityequivchartwoevaluation} if we can compute $P_{n,k,s,\FF_{2^{\kappa}}}$ for worst-case inputs, then we can solve \textsc{Parity-$(k,s)$-clique}. We use Lemma \ref{lem:reedsolomon} to construct $A_0$ from $A_1$ such that $A_0$ runs in time $$T(A_0,n) = O(\kappa^2 (N+D^2)D + D \cdot T(A_1,n)) = O(n^s \binom{k}{s}^2 \kappa^2 + \binom{k}{s} \cdot T(A_1,n))$$

\textit{3. Reduce to computing $P_{n,k,s,\FF_2}$ on random inputs in $\FF_2^N$.} We use Lemma \ref{lem:fpttofp} to construct $A_1$ from $A_2$ such that $A_1$ runs in time $$T(A_1,n) \leq O(N\kappa^4 + \kappa^{\binom{k}{s}} \cdot T(A_2,n)),$$ and has error probability at most $\delta_2 \cdot \kappa^{\binom{k}{s}}$ on random inputs $w \sim \mathrm{Unif}[\FF_{2^{\kappa}}^N]$. Thus, $A_2$ must have error probability at most $\delta_2 < 1/(3\kappa^{\binom{k}{s}})$ on random inputs in $y \sim \mathrm{Unif}[\FF_2^N]$ for this step of the reduction to work.

\textit{4. Reduce to computing $P_{n,k,s,\FF_2}$ on random inputs in $\{0,1\}^N$.} We use Lemma \ref{lem:fptogncsk} to construct $A_2$ from $A_3$ such that $A_2$ runs in time $$T(A_2,n) = O(N t \log (N/\gamma) + t^{\binom{k}{s}} \cdot T(A_3,n)),$$ for some $t = O(c^{-1}(1-c)^{-1} (s \log(n) + \log(1/\gamma))).$ The error probability of $A_2$ on random inputs $z \sim G(nk,c,s,k)$ will be at most $\delta_2 < \delta_3 \cdot t^{\binom{k}{s}} + \gamma$. Since we require error probability at most $\delta_2 \leq 1/(3\kappa^{\binom{k}{s}})$ of algorithm $A_2(z)$ on inputs $z \sim G(nk,c,s,k)$, we set $\gamma = 1/(10 \kappa^{\binom{k}{s}})$ and require $\delta_3 \leq 1/(10 (t \kappa)^{\binom{k}{s}}),$ which is sufficient. For this choice of $\gamma$, we have $t = O(c^{-1}(1-c)^{-1} (s \log(n) + \binom{k}{s} \log \log \binom{k}{s}))$.

\textit{5. Reduce to computing \textsc{Parity-$(k,s)$-clique} for $G(n,c,s)$ hypergraphs.}
We use Lemma \ref{lem:erkpartitetoergeneralreductionparity} to construct $A_3$ from $A$ such that $A_3$ runs in time $$T(A_3,n) = O((nk)^s + s^2 k^3 2^k \log^2(nk) + 2^k \cdot T(A,nk)),$$ and such that $A_3$ has error probability at most $\delta_3 < 2^k / \Upsilon_{P,1}$.
\\

As in the theorem statement, let
$$\Upsilon_{P,1}(n,c,s,k) \triangleq \left(C (c^{-1}(1-c)^{-1}) s (\log k)\left(s\log n + \binom{k}{s} \log \log \binom{k}{s}\right)\right)^{\binom{k}{s}}$$ for some large enough constant $C$.

If we take $C$ large enough, then $(\kappa t)^{\binom{k}{s}} \leq \frac{1}{10} \cdot 2^{-k} \cdot \Upsilon_{P,1}$, as desired. In this case, the error of $A_0$ on uniformly random inputs will be at most $1/3$, which is what we needed. Putting the runtime bounds together,
\begin{align*}
T(B,n) &= O\Big(n^s \binom{k}{s}^2 \kappa^2 + n^s \binom{k}{s}^2 \kappa^{\binom{k}{s}} t \log\left(n^s \kappa^{\binom{k}{s}}\right) \\
&\quad \quad + n^s \binom{k}{s}^2 \kappa^4 + \binom{k}{s} \cdot (4\kappa t)^{\binom{k}{s}} \cdot (T(A,nk) + (nk)^s) \Big) \\
&= O(n^s \binom{k}{s}^2(\latestedits{(\log n)} \cdot t k  \kappa^{\binom{k}{s}}\log^2 \kappa + \kappa^4) + \Upsilon_{P,1} \cdot (T(A,nk) + (nk)^s)),
\end{align*}
if we choose $C > 0$ large enough. Since $\binom{k}{s} \geq k \geq 3$ without loss of generality, the second term dominates and $$T(B,n) = O(\Upsilon_{P,1} \cdot (T(A,nk) + (nk)^s)).$$

For item 2 of the theorem, we restrict the inputs to come from $G(n,1/2,s)$, and we achieve a better error tolerance because algorithm $A_3$ is the same as $A_2$. This means that we may skip step 4 of the proof of item 1. In particular, we only need $\delta_3 = \delta_2 \leq 1 / (3 \kappa^{\binom{k}{s}})$. So algorithm $A$ only needs to have error $< 1 / \Upsilon_{P,2}$, for $\Upsilon_{P,2}(k,s) \triangleq (C s \log k)^{\binom{k}{s}}$. It is not hard to see that, skipping step 4, the algorithm $B$ that we construct takes time $T(B,n) = O(\Upsilon_{P,2} \cdot (T(A,nk) + (nk)^s))$.
\end{proof}

\section{Random Binary Expansions Modulo $p$}
\label{sec:randombinaryexpansions}

\latestedits{We fix some notation to be used throughout this section. Let $p$ be a prime number, let $c \in (0,1/2]$, and let $q_0,\ldots,q_t \in [c,1-c]$ be probabilities. Let $Z = (Z^{(0)},\ldots,Z^{(t)}) \in \{0,1\}^{t+1}$ be a vector of independent, biased Bernoulli random variables such that $Z^{(i)} \sim \Ber(q_i)$ for all $i \in \{0,\ldots,t\}$.} In this section, we consider the distributions of random binary expansions modulo $p$, of the form
$$\latestedits{Z^{(t)} \cdot 2^t + Z^{(t - 1)} \cdot 2^{t - 1} + \cdots + Z^{(0)}} \pmod{p}.$$
We show that for $t$ polylogarithmic in $p$, these distributions become close to uniformly distributed over $\mathbb{F}_p$. This is then used to go in the other direction, producing approximately independent Bernoulli variables that are the binary expansion of a number with a given residue. The special case of this argument in which the Bernoulli variables are unbiased has already appeared in an earlier work by Goldreich and Rothblum \cite{goldreich2017worst}. In that case, the proof of correctness is much simpler, because the Fourier-analytic tools used below can be avoided.

For $p > 2$, the main result of the section is the following slightly more general restatement of Lemma~\ref{lem:samplefrommodp}. It implies that we can efficiently sample biased binary expansions, conditioned on the expansion being equivalent to some $x$ modulo $p$.
\begin{lemma}[Restatement of Lemma~\ref{lem:samplefrommodp}] \label{lem:samplefrommodprestated}
\latestedits{There is an absolute constant $K > 0$ such that the following holds. Let $p > 2$ be prime, $\eps > 0$, and $t \geq K \cdot c^{-1} (1-c)^{-1} \log (p/\epsilon) \log p$. Then there is an $O(pt \log(1/\eps)\log(p))$-time randomized algorithm that, given $x \in \FF_p$, outputs $\tilde{Z}_x = (\tilde{Z}_x^{(0)},\ldots,\tilde{Z}_x^{(t)}) \in \{0,1\}^{t+1}$ satisfying  $\sum_{i=0}^{t} 2^i \cdot \tilde{Z}_{x}^{(i)} \equiv x \pmod{p}$ almost surely. Moreover, if $R \sim \mathrm{Unif}[\FF_p]$ then $\dtv(\cL(\tilde{Z}_R),\cL(Z)) < \eps$.}
\end{lemma}

Our argument uses finite Fourier analysis on $\mathbb{F}_p$. Given a function $f : \FF_p \to \RR$, define its Fourier transform to be $\hat{f} : \FF_p \to \CC$, where $\hat{f}(t) = \sum_{x=0}^{p-1} f(x) \omega^{tx}$ and $\omega = e^{2\pi i / p}$. In this section, we endow $\mathbb{F}_p$ with the total ordering of $\{0, 1, \dots, p - 1\}$ as elements of $\mathbb{Z}$. Given a set $S$, let $2S = \{ 2s : s \in S\}$. We begin with a simple claim showing that sufficiently long geometric progressions with ratio 2 in $\mathbb{F}_p$ contain a middle residue modulo $p$.

\begin{claim} \label{claim:geoprog}
Suppose that $a_1,\ldots,a_k \in \mathbb{F}_p$ is a sequence with $a_1 \neq 0$ and $a_{i+1} = 2a_i$ for each $1 \leq i \leq k-1$. Then if $k \ge 1 + \log_2(p/3)$, there is some $j$ with $\frac{p}{3} \le a_j \le \frac{2p}{3}$.
\end{claim}

\begin{proof}
Let $S = \{x \in \FF_p : x < p/3\}$ and $T = \{x \in \FF_p : x > 2p/3\}$. Observe that $2S \cap T = \emptyset$ and $S \cap 2T = \emptyset$, which implies that there is no $i$ such that $a_i$ and $a_{i + 1}$ are both in $S$ and $T$. Therefore if $(a_1, a_2, \dots, a_k)$ contains elements of both $S$ and $T$, there must be some $j$ with $a_j \in (S \cup T)^C$ and the claim follows. It thus suffices to shows that $(a_1, a_2, \dots, a_k)$ cannot be entirely contained in one of $S$ or $T$. First consider the case that it is contained in $S$. Define the sequence $(a_1', a_2', \dots, a_k')$ of integers by $a'_{i+1} = 2a'_i$ for each $1 \le i \le k - 1$ and $a_1' \in [1, p/3)$ is such that $a_1' \equiv a_1 \pmod{p}$. It follows that $a_i' \equiv a_i \pmod{p}$ for each $i$ and $a_k' \ge 2^{k - 1} \ge p/3$. Now consider the smallest $j$ with $a_j' > p/3$. Then $p/3 \ge a'_{j-1} = a'_j/2$ by the minimality of $j$, and $p/3 \le a_j \le 2p/3$ which is a contradiction. If the sequence is contained in $T$, then $(-a_1,-a_2,\ldots,-a_k)$ is contained in $S$ and applying the same argument to this sequence proves the claim.
\end{proof}

We now bound the total variation between the distribution of random binary expansions modulo $p$ and the uniform distribution. In Appendix~\ref{app:binarytightness}, we show Lemma~\ref{lem:tvislowfourier} is tight assuming there are infinitely-many Mersenne primes.

\begin{lemma} \label{lem:tvislowfourier}
\latestedits{There is an absolute constant $K > 0$ such that the following holds. Let $p > 2$ be prime, $\epsilon > 0$, and $t \ge K \cdot c^{-1} (1-c)^{-1} \log (p/\epsilon) \log p$. Define the random variable $Y = \sum_{i=0}^{t} 2^i \cdot Z^{(i)} \in \{0,\ldots,2^{t+1}-1\},$ and define the random variable $X \in \FF_p$ by $X \equiv Y \pmod{p}$. Then, letting $\mathcal{L}(X)$ denote the law of $X$, we have $$\dtv(\cL(X),\mathrm{Unif}[\FF_p]) \leq \epsilon.$$}
\end{lemma}

\begin{proof}
\latestedits{Let $f : \FF_p \to \RR$ be the probability mass function of $X$.} By definition, we have that
$$f(x) = \sum_{z \in \{0, 1\}^{t+1}} \left( \prod_{i = 0}^t q_i^{z_i} (1 - q_i)^{1 - z_i} \right) \mathbf{1} \left\{ \sum_{i = 0}^t   2^i\cdot z_i \equiv x \pmod{p} \right\}$$
\latestedits{Now observe that $\hat{f}(s)$ is given by
\begin{align*}\hat{f}(s) &= \sum_{x=0}^{p-1} f(x) \omega^{sx} = \sum_{z \in \{0,1\}^{t+1}} \left( \prod_{i=0}^t q_i^{z_i} (1-q_i)^{1-z_i}\right) \left(\omega^{s \sum_{i=0}^t 2^i \cdot z_i}\right) \\ &= \prod_{i = 0}^{t} \left( 1 - q_i + q_i \cdot \omega^{2^i \cdot s} \right)\end{align*}
The last equality follows directly from expanding the product.} Note that the constant function $\mathbf{1}$ has Fourier transform $p \cdot \mathbf{1}_{\{ s = 0\}}$. By Cauchy-Schwarz and Parseval's theorem, we have that
\begin{align*}
4 \cdot \dtv\left( \latestedits{\mathcal{L}(X)}, \text{Unif}[\mathbb{F}_p] \right)^2 &= \| f - p^{-1} \cdot \mathbf{1} \|_1^2 \le p \cdot \| f - p^{-1} \cdot \mathbf{1} \|_2^2 = \| \hat{f} - \mathbf{1}_{\{ s = 0\}} \|_2^2 \\
&= \sum_{s \neq 0} \prod_{i = 0}^{t} \left| 1 - q_i + q_i \cdot \omega^{2^i \cdot s} \right|^2.
\end{align*} Note that $|1 - q + q \cdot \omega^a| \le 1$ by the triangle inequality for all $a \in \mathbb{F}_p$ and $q \in (0, 1)$. Furthermore, if $a \in \mathbb{F}_p$ is such that $p/3 \le a \le 2p/3$ and $q \in [c, 1 - c]$, then we have that
\begin{align*}
\left|1 - q + q \cdot \omega^a\right|^2 &= (1 - q)^2 + q^2 + 2q(1 - q) \cos(2\pi a/p) \\
&= 1 - 2q(1 - q) \left( 1 - \cos(2\pi a /p) \right) \\
&\le 1 - 2c(1 - c) \left( 1 - \cos(4\pi/3) \right) \\
&= 1 - 3c(1 - c)
\end{align*}
since $\cos(x)$ is maximized at the endpoints on the interval $x \in [2\pi/3, 4\pi/3]$ and $q(1 - q)$ is minimized at the endpoints on the interval $[c, 1 - c]$. Now suppose that $t$ is such that
$$t \ge \left\lceil \frac{\log(4\epsilon^2/p)}{\log(1 - 3c(1 - c))} \right\rceil \cdot \left\lceil 1 + \log_2(p/3) \right\rceil = \Theta\left( c^{-1} (1-c)^{-1} \log (p/\epsilon) \log p \right)$$
Fix some $s \in \mathbb{F}_p$ with $s \neq 0$. By Claim \ref{claim:geoprog}, any $\left\lceil 1 + \log_2(p/3) \right\rceil$ consecutive terms of the sequence $s, 2s, \dots, 2^t s \in \mathbb{F}_p$ contain an element between $p/3$ and $2p/3$. Therefore this sequence contains at least $m = \left\lceil \frac{\log(4\epsilon^2/p)}{\log(1 - 3c(1 - c))} \right\rceil$ such terms, which implies that
$$\prod_{i = 0}^{t} \left| 1 - q_i + q_i \cdot \omega^{2^i \cdot s} \right|^2 \le \left( 1 - 3c(1 - c) \right)^{m} \le \frac{4\epsilon^2}{p}$$
by the inequality above and the fact that each term in this product is at most $1$. Since this holds for each $s \neq 0$, it now follows that
$$4 \cdot \dtv\left( \latestedits{\mathcal{L}(X)}, \text{Unif}[\mathbb{F}_p] \right)^2 \le \sum_{s \neq 0} \prod_{i = 0}^{t} \left| 1 - q_i + q_i \cdot \omega^{2^i \cdot s} \right|^2 < 4\epsilon^2$$
and thus $\dtv\left( \mathcal{L}(\latestedits{X}), \text{Unif}[\mathbb{F}_p] \right) < \epsilon$, proving the lemma.
\end{proof}

\latestedits{Using the above lemma, we can now prove the main result of this section for $p > 2$. The idea is to rejection sample $Z = (Z^{(0)},\ldots,Z^{(t)})$ conditioned on $\latestedits{X} \equiv x \pmod{p}$.}

\begin{proof}[Proof of Lemma~\ref{lem:samplefrommodprestated}]
\latestedits{Define the random variable $Y = \sum_{i=0}^{t} 2^i \cdot Z^{(i)} \in \{0,\ldots,2^{t+1}-1\},$ and define the random variable $X \in \FF_p$ by $X \equiv Y \pmod{p}$, as in Lemma~\ref{lem:tvislowfourier}. Let $K > 0$ be large enough that, by Lemma~\ref{lem:tvislowfourier}, we have $$\dtv\left(\mathcal{L}(X), \textnormal{Unif}[\FF_p]\right) < \epsilon/(2p).$$ We sample a random variable $\tilde{Y}_{x} \in \{0,\ldots,2^{t+1}-1\}$ by rejection sampling from the distribution $\mathcal{L}(Y)$ until receiving an element congruent to $x$ modulo $p$ or reaching the cutoff of
$$m = \left\lceil \frac{\log (\epsilon/2)}{\log(1 - 1/(2p))} \right\rceil = O\left(p\log(1/\epsilon) \right)$$
rounds, in which case we stop and set $\tilde{Y}_{x}$ to an arbitrary value congruent to $x$. We then return  $\tilde{Z}_x = (\tilde{Z}_{x}^{(0)},\ldots,\tilde{Z}_{x}^{(t)})$, the binary expansion of $\tilde{Y}_{x}$ from lowest-order bit to highest-order bit.

By construction, it holds that $\sum_{i=0}^t 2^i \cdot \tilde{Z}_{x}^{(i)} = \tilde{Y}_x \equiv x \pmod{p}$ almost surely. Furthermore, the runtime bound follows because each sample from $\mathcal{L}(Y)$ can be obtained in $O(t)$ time by sampling $Z^{(0)}, Z^{(1)}, \dots, Z^{(t)}$ and forming the number with binary digits $Z^{(t)}, Z^{(t-1)}, \dots, Z^{(0)}$. Checking whether this number is congruent to $x$ modulo $p$ takes $O(t \log(p))$ time by Theorem 3.3 of \cite{shoup2009computational}.

It remains to prove that $(\tilde{Z}_{x}^{(0)},\ldots,\tilde{Z}_{x}^{(t)})$ is close to $(Z^{(0)},\ldots,Z^{(t)})$ in total variation if $x$ is chosen uniformly in $\FF_p$. We begin by considering the case of fixed $x \in \FF_p$. Let $Y_{x}$ be a random variable with the conditional law $\cL(Y_x) \triangleq \cL(Y | Y \equiv x \pmod{p})$. If we receive a sample from $\cL(Y)$ congruent to $x$ by the $m$th round of rejection sampling, then it is exactly sampled from $\cL(Y_x)$. Therefore $\dtv(\cL(\tilde{Y}_x), \cL(Y_x))$ is upper bounded by the probability that the rejection sampling scheme fails to output a sample. Now note that the probability that a sample is output in a single round is
$$\mathbb{P}[X = x] \ge 1/p - \dtv\left(\mathcal{L}(X), \textnormal{Unif}[\FF_p]\right) > 1/p - \eps/(2p) \geq 1/(2p)$$
by the definition of total variation. By the independence of sampling in different rounds, the probability that no sample is output is at most
$$\left(1 - \mathbb{P}[X = x] \right)^m \le \left( 1 - 1/(2p) \right)^m \le \epsilon/2.$$
So we may conclude that, for any fixed $x \in \FF_p$, \begin{equation*}\dtv(\cL(\tilde{Y}_x),\cL(Y_x)) \leq \epsilon/2.\end{equation*}

Now, let $R \sim \mathrm{Unif}[\FF_p]$. By the above inequality, we have \begin{align} \dtv(\cL(\tilde{Y}_R), \cL(Y_{R})) \leq \frac{1}{p} \sum_{x \in \FF_p} \dtv(\cL(\tilde{Y}_x), \cL(Y_x)) \leq \eps/2. \label{ineq:dtvtildevsnotilde}\end{align}
We now bound the total variation distance between $\cL(Y_{R})$ and $\cL(Y)$. Let $X' \sim \cL(X)$ be independent of the other variables, and note that $\cL(Y_{X'}) = \cL(Y)$, since, for any $y \in \{0,\ldots,2^{t+1}-1\}$, Bayes' rule implies \begin{align*}\PP(Y_{X'} = y) &= \PP(Y_{X'} = y \mid Y_{X'} \equiv y \pmod{p}) \cdot \PP(Y_{X'} \equiv y \pmod{p}) \\ 
&= \PP(Y_{X'} = y \mid X' \equiv y \pmod{p}) \cdot \PP(X' \equiv y \pmod{p})\\
&= \PP(Y = y \mid Y \equiv y \pmod{p}) \cdot \PP(Y \equiv y \pmod{p}) \\
&= \PP(Y = y).\end{align*}
So by the data processing inequality, since $x \mapsto Y_x$ is a Markov transition sending $R$ to $Y_{R}$ and $X'$ to $Y_{X'}$,
\begin{align}\dtv(\cL(Y),\cL(Y_{R})) &= \dtv(\cL(Y_{X'}), \cL(Y_{R})) \nonumber \\ &\leq \dtv(\cL(X'),\cL(R)) = \dtv(\cL(X),\mathrm{Unif}[\FF_p]) < \eps/2\label{eq:dtvsamplebound2}.\end{align}
Finally, since $(\tilde{Z}_{R}^{(t)},\ldots,\tilde{Z}_{R}^{(0)})$ is the binary expansion of $\tilde{Y}_R$, and $(Z^{(t)},\ldots,Z^{(0)})$ is the binary expansion of $Y$, the data processing inequality implies:  \begin{equation}\dtv(\cL(\tilde{Z}_{R}^{(0)},\ldots,\tilde{Z}_{R}^{(t)}),\cL(Z^{(0)},\ldots,Z^{(t)})) \leq \dtv(\mathcal{L}(\tilde{Y}_R),\mathcal{L}(Y)).\label{eq:dtvsamplebound1}
\end{equation}
We bound the right-hand-side of \eqref{eq:dtvsamplebound1} with the triangle inequality, \eqref{ineq:dtvtildevsnotilde}, and \eqref{eq:dtvsamplebound2}:
\begin{align*}\dtv(\cL(\tilde{Z}_{R}^{(0)},&\ldots,\tilde{Z}_{R}^{(t)}),\cL(Z^{(0)},\ldots,Z^{(t)})) \\ &\leq \dtv(\cL(\tilde{Y}_R),\cL(Y_{R})) + \dtv(\cL(Y_{R}),\cL(Y)) < \eps/2 + \eps/2 = \eps.\end{align*}}
\end{proof}

We conclude with a sampling result analogous to Lemma \ref{lem:samplefrommodprestated}, but for $p = 2$.

\begin{lemma}[Sampling lemma for $p = 2$]\label{lem:samplefrommodtwo}
\latestedits{There is a constant $K > 0$ such that the following holds. Let $\eps > 0$ and $t \geq Kc^{-1}(1-c)^{-1} \log (1/\eps)$. Then there is an $O(t \log(1/\eps))$-time randomized algorithm that, given $x \in \FF_2$, outputs $\tilde{Z}_{x} = (\tilde{Z}_{x}^{(0)},\ldots,\tilde{Z}_{x}^{(t)}) \in \{0,1\}^{t+1}$ satisfying  $\sum_{i=0}^{t} \tilde{Z}_{x}^{(i)} \equiv x \pmod{2}$ almost surely.
Moreover, if $R \sim \mathrm{Unif}[\FF_2]$ then $\dtv(\cL(\tilde{Z}_{R}), \cL(Z)) < \eps$.}
\end{lemma}
\begin{proof}
By induction on $t$, one may show that
$$\PP\left[\sum_{i=0}^t \latestedits{Z^{(i)}} \equiv 0 \pmod{2}\right] = \frac{1}{2} \latestedits{+} \frac{\prod_{i=0}^t (1-2q_i)}{2}$$
\latestedits{If $t$ satisfies the lower bound $t \geq \ceil{\log(\eps/4)/\log(|1-2c|)} + 1 = \latestedits{O}(c^{-1}(1-c)^{-1} \log(1/\eps))$, it holds that $\dtv(\cL(\sum_{i=0}^t Z^{(i)} \pmod{2}), \cL(R)) < \min(1/4,\eps/2)$.

The proof now proceeds analogously to the proof of Lemma~\ref{lem:samplefrommodprestated}. We sample $\tilde{Z}_{x} = (\tilde{Z}_{x}^{(0)},\ldots,\tilde{Z}_{x}^{(t)})$ by rejection sampling from $\cL(Z)$ until receiving a vector whose sum is is congruent to $x$ modulo $2$, or cutting off at $\Theta(\log(1/\eps))$ rounds. This takes $O(t \log(1/\eps))$ time, because it consists of at most $O(\log(1/\eps))$ rounds of sampling fresh copies of $Z^{(i)} \sim \Ber(q_i)$ for all $i \in \{0,\ldots,t\}$ and checking if $\sum_{i=0}^t Z^{(i)} \equiv x \pmod{2}$. Let $Z_x$ be a random variable with the conditional law $\cL(Z_x) \triangleq \cL(Z \mid Z \equiv x \pmod{2})$. Then the rejection sampling outputs $\tilde{Z}_{x}$ satisfying 
$\dtv(\cL(\tilde{Z}_{x}),\cL(Z_x)) \leq \eps/2$, so \begin{equation}\dtv(\cL(\tilde{Z}_{R}),\cL(Z_{R})) \leq \eps/2.\label{ineq:dtvpeq2part1}\end{equation}
Further, by applying the data processing inequality with Markov kernel $x \mapsto Z_x$, with reasoning analogous to the proof of \eqref{eq:dtvsamplebound2}, we derive \begin{equation}\textstyle \dtv(\cL(Z_{R}),\cL(Z)) \leq \dtv(\cL(R),\cL(\sum_{i=0}^t Z^{(i)} \pmod{2})) < \eps/2.\label{ineq:dtvpeq2part2}\end{equation} Combining \eqref{ineq:dtvpeq2part1} and \eqref{ineq:dtvpeq2part2} with triangle inequality yields $\dtv(\cL(\tilde{Z}_{R}),\cL(Z)) < \eps$.}
\end{proof}

\section{Algorithms for Counting $k$-Cliques in $G(n, c, s)$}\label{sec:algs}

In this section, we consider several natural algorithms for counting $k$-cliques in $G(n, c, s)$ with $c = \Theta(n^{-\alpha})$ for some $\alpha \in (0, 1)$. The main objective of this section is to show that, when $k$ and $s$ are constant, these algorithms all run faster than all known algorithms for $\#(k, s)$-\textsc{clique} on worst-case hypergraphs and nearly match the lower bounds from our reduction for certain $k$, $c$ and $s$. This demonstrates that the average-case complexity of $\#(k, s)$-\textsc{clique} on Erd\H{o}s-R\'{e}nyi hypergraphs is intrinsically different from its worst-case complexity. As discussed in Section \ref{subsec:timeslowdownnecessity}, this also shows the necessity of a slowdown term comparable to $\Upsilon_{\#}$ in our worst-case to average-case reduction for $\#(k, s)$-\textsc{clique}. We begin with a randomized sampling-based algorithm for counting $k$-cliques in $G(n, c, s)$, extending well-known greedy heuristics for finding $k$-cliques in random graphs. We then present an improvement to this algorithm in the graph case and a deterministic alternative.

\subsection{Greedy Random Sampling}
\label{subsec:greedysampling}

In this section, we consider a natural greedy algorithm $\textsc{greedy-random-sampling}$ for counting $k$-cliques in a $s$-uniform hypergraph $G \sim G(n, c, s)$ with $c = \Theta(n^{-\alpha})$. Given a subset of vertices $A \subseteq [n]$ of $G$, define $\textsc{cn}_G(A)$ to be
$$\textsc{cn}_G(A) = \left\{ v \in V(G) \backslash A : B \cup \{ v \} \in E(G) \text{ for all } (s - 1)\text{-subsets } B \subseteq A \right\}$$
or, in other words, the set of common neighbors of the vertices in $A$. The algorithm $\textsc{greedy-random-sampling}$ maintains a set $S$ of $k$-subsets of $[n]$ and for $T$ iterations does the following:
\begin{enumerate}
\item Sample distinct starting vertices $v_1, v_2, \dots, v_{s-1}$ uniformly at random and proceed to sample the remaining vertices $v_s, v_{s+1}, \dots, v_k$ iteratively such that $v_{i+1}$ is chosen uniformly at random from $\textsc{cn}_G(v_1, v_2, \dots, v_i)$ if it is nonempty.
\item If $k$ vertices $\{v_1, v_2, \dots, v_k\}$ are chosen then add $\{v_1, v_2, \dots, v_k\}$ to $S$ if it is not already in $S$.
\end{enumerate}
This algorithm is an extension of the classical greedy algorithm for finding $\log_2 n$ sized cliques in $G(n, 1/2)$ in \cite{karp1976, grimmett1975colouring}, the Metropolis process examined in \cite{jerrum1992large} and the greedy procedure solving $k$-\textsc{clique} on $G(n, c)$ with $c = \Theta\left(n^{-2/(k-1)}\right)$ discussed by Rossman in \cite{rossman2016lower}. These and other natural polynomial time search algorithms fail to find cliques of size $(1 + \epsilon) \log_2 n$ in $G(n, 1/2)$, even though its clique number is approximately $2 \log_2 n$ with high probability \cite{mcdiarmid1984colouring, pittel1982probable}. Our algorithm $\textsc{greedy-random-sampling}$ extends this greedy algorithm to count $k$-cliques in $G(n, c, s)$. In our analysis, we will see a phase transition in the behavior of this algorithm at $k = \tau$ for some $\tau$ smaller than the clique number of $G(n, c, s)$. This is analogous to the breakdown of the natural greedy algorithm at cliques of size $\log_2 n$ on $G(n, 1/2)$.

Before analyzing $\textsc{greedy-random-sampling}$, we state a simple classical lemma counting the number of $k$-cliques in $G(n, c, s)$. This lemma follows from linearity of expectation and Markov's inequality. Its proof is included in Appendix \ref{sec:cliquecounts} for completeness.

\begin{lemma} \label{lem:cliquenumber}
For fixed $\alpha \in (0, 1)$ and $s$, let $\kappa \ge s$ be the largest positive integer satisfying $\alpha \binom{\kappa}{s - 1} < s$. If $G \sim G(n, c, s)$ where $c = O(n^{-\alpha})$, then $\mathbb{E}[|\cl_k(G)|] = \binom{n}{k} c^{\binom{k}{s}}$ and $\omega(G) \le \kappa + 1 + t$ with probability at least $1 - O\left(n^{-\alpha t(1 - s^{-1}) \binom{\kappa + 2}{s - 1}}\right)$ for any fixed nonnegative integer $t$, where the constant in the $O(\cdot)$ notation can depend on $t$.
\end{lemma}

In particular, this implies that the clique number of $G(n, c, s)$ is typically at most $(s! \alpha^{-1} )^{\frac{1}{s-1}} + s$. In the graph case of $s = 2$, this simplifies to $2\alpha^{-1} + 2$. In the next subsection, we give upper bounds on the number of iterations $T$ causing all $k$-cliques in $G$ to end up in $S$ and analyze the runtime of the algorithm. The subsequent subsection improves the runtime of $\textsc{greedy-random-sampling}$ for graphs when $s = 2$ through a matrix multiplication post-processing step. The last subsection gives an alternative deterministic algorithm with a similar performance to $\textsc{greedy-random-sampling}$.

\subsection{Sample Complexity and Runtime of Greedy Random Sampling}

In this section, we analyze the runtime of $\textsc{greedy-random-sampling}$ and give upper bounds on the number of iterations $T$ needed for the algorithm to terminate with $S = \cl_k(G)$. The dynamic set $S$ needs to support search and insertion of $k$-cliques. Consider labelling the vertices of $G$ with elements of $[n]$ and storing the elements of $S$ in a balanced binary search tree sorted according to the lexicographic order on $[n]^k$. Search and insertion can each be carried out in $O(\log |\cl_k(G)|) = O(k \log n)$ time. It follows that each iteration of $\textsc{greedy-random-sampling}$ therefore takes $O(kn + k \log n) = O(n)$ time as long as $k = O(1)$. Outputting $|S|$ in $\textsc{greedy-random-sampling}$ therefore yields a $O(nT)$ time algorithm for $\#(k, s)$-\textsc{clique} on $G(n, c, s)$ that succeeds with high probability.

The following theorem provides upper bounds on the minimum number of iterations $T$ needed for this algorithm to terminate with $S = \cl_k(G)$ and therefore solve $\#(k, s)$-\textsc{clique}. Its proof is deferred to Appendix \ref{sec:grs-analysis}.

\begin{theorem} \label{thm:greedysample}
Let $k$ and $s$ be constants and $c = \Theta(n^{-\alpha})$ for some $\alpha \in (0, 1)$. Let $\tau$ be the largest integer satisfying $\alpha \binom{\tau}{s - 1} < 1$ and suppose that
$$T \ge \left\{ \begin{matrix} 2n^{\tau + 1} c^{\binom{\tau + 1}{s}} (3 \log n)^{(k - \tau) (1 + \epsilon)} & \text{if } k \ge \tau + 1 \\ 2n^{k} c^{\binom{k}{s}} (\log n)^{1 + \epsilon} & \text{if } k < \tau + 1 \end{matrix} \right.$$
for some $\epsilon > 0$. Then $\textsc{greedy-random-sampling}$ run with $T$ iterations terminates with $S = \cl_k(G)$ with probability $1 -n^{-\omega(1)}$ over the random bits of the algorithm $\textsc{greedy-random-sampling}$ and over the choice of random hypergraph $G \sim G(n, c, s)$.
\end{theorem}

Implementing $S$ as a balanced binary search tree and outputting $|S|$ in \textsc{greedy-random-sampling} yields the following algorithmic upper bounds for $\#(k, s)$-\textsc{clique} with inputs sampled from $G(n, c, s)$.

\begin{corollary}
Suppose that $k$ and $s$ are constants and $c = \Theta(n^{-\alpha})$ for some $\alpha \in (0, 1)$. Let $\tau$ be the largest integer satisfying $\alpha \binom{\tau}{s - 1} < 1$. Then it follows that
\begin{enumerate}
\item If $k \ge \tau + 1$, there is an $\tilde{O}\left(n^{\tau+2 - \alpha \binom{\tau + 1}{s}}\right)$ time randomized algorithm solving $\#(k, s)$-\textsc{clique} on inputs sampled from $G(n, c, s)$ with probability at least $1 - n^{-\omega(1)}$.
\item If $k < \tau + 1$, there is an $\tilde{O}\left(n^{k + 1 - \alpha \binom{k}{s}}\right)$ time randomized algorithm solving $\#(k, s)$-\textsc{clique} on inputs sampled from $G(n, c, s)$ with probability at least $1 - n^{-\omega(1)}$.
\end{enumerate}
\end{corollary}

By Lemma \ref{lem:cliquenumber}, the hypergraph $G \sim G(n, c, s)$ has clique number $\omega(G) \le \kappa + 2$ with probability $1 - 1/\text{poly}(n)$ where $\kappa \ge s$ is the largest positive integer satisfying $\alpha \binom{\kappa}{s - 1} < s$. In particular, when $k > \kappa + 2$ in the theorem above, the algorithm outputting zero succeeds with probability $1 - 1/\text{poly}(n)$ and $\#(k, s)$-\textsc{clique} is trivial. For there to typically be a nonzero number of $k$-cliques in $G(n, c, s)$, it should hold that $0 < \alpha \le s \binom{k - 1}{s - 1}^{-1}$. In the graph case of $s = 2$, this simplifies to the familiar condition that $0 < \alpha \le \frac{2}{k - 1}$. We also remark that when $k < \tau + 1$, the runtime of this algorithm is an $\tilde{O}(n)$ factor off from the expectation of the quantity being counted, the number of $k$-cliques in $G \sim G(n, c, s)$.

\subsection{Post-Processing with Matrix Multiplication}

In this section, we improve the runtime of $\textsc{greedy-random-sampling}$ as an algorithm for $\#(k, s)$-\textsc{clique} in the graph case of $s = 2$. The improvement comes from the matrix multiplication step of Ne\u{s}et\u{r}il and Poljak from their $O\left(n^{\omega \lfloor k/3 \rfloor + (k \pmod{3})}\right)$ time worst-case algorithm for $\#(k, 2)$-\textsc{clique} \cite{nevsetvril1985complexity}. Our improved runtime for the algorithm $\textsc{greedy-random-sampling}$ is stated in the following theorem.

\begin{theorem} \label{thm:matrixaug}
Suppose that $k > 2$ is a fixed positive integer and $c = \Theta(n^{-\alpha})$ where $0 < \alpha \le \frac{2}{k - 1}$ is also fixed. Then there is a randomized algorithm solving $\#(k, 2)$-\textsc{clique} on inputs sampled from $G(n, c)$ with probability $1 - n^{-\omega(1)}$ that runs in $\tilde{O}\left( n^{\omega \lceil k/3 \rceil + \omega - \omega \alpha \binom{\lceil k/3 \rceil}{2}} \right)$ time.
\end{theorem}

\begin{proof}
Label the vertices of an input graph $G \sim G(n, c)$ with the elements of $[n]$. Consider the following application of $\textsc{greedy-random-sampling}$ with post-processing:
\begin{enumerate}
\item Run $\textsc{greedy-random-sampling}$ to compute the two sets of cliques $S_1 = \cl_{\lfloor k/3 \rfloor}(G)$ and $S_2 = \cl_{\lceil k/3 \rceil}(G)$ with the number of iterations $T$ as given in Theorem \ref{thm:greedysample}.
\item Construct the matrix $M_1 \in \{0, 1\}^{|S_1| \times |S_1|}$ with rows and columns indexed by the elements of $S_1$ such that $(M_1)_{A, B} = 1$ for $A, B \in S_1$ if $A \cup B$ forms a clique of $G$ and all labels in $A$ are strictly less than all labels in $B$.
\item Construct the matrix $M_2 \in \{0, 1\}^{|S_1| \times |S_2|}$ with rows indexed by the elements of $S_1$ and columns indexed by the elements of $S_2$ such that $(M_2)_{A, B} = 1$ for $A \in S_1$ and $B \in S_2$ under the same rule that $A \cup B$ forms a clique of $G$ and all labels in $A$ are strictly less than all labels in $B$. Construct the matrix $M_3$ with rows and columns indexed by $S_2$ analogously.
\item Compute the matrix product
$$M_P = \left\{ \begin{matrix} M_1^2 &\text{if } k \equiv 0 \pmod{3} \\ M_1 M_2 &\text{if } k \equiv 1 \pmod{3} \\ M_2 M_3 &\text{if } k \equiv 2 \pmod{3} \end{matrix} \right.$$
\item Output the sum of entries
$$\sum_{(A, B) \in \mathcal{S}} (M_P)_{A, B}$$
where $\mathcal{S}$ is the support of $M_1$ if $k \equiv 0 \pmod{3}$ and $\mathcal{S}$ is the support of $M_2$ if $k \not \equiv 0 \pmod{3}$.
\end{enumerate}
We will show that this algorithm solves $\#(k, 2)$-\textsc{clique} with probability $1 - n^{-\omega(1)}$ when $k \equiv 1 \pmod{3}$. The cases when $k \equiv 0, 2 \pmod{3}$ follow from a nearly identical argument. By Theorem \ref{thm:greedysample}, the first step applying $\textsc{greedy-random-sampling}$ succeeds with probability $1 - n^{-\omega(1)}$. Note that $(M_P)_{A, B}$ counts the number of $\lfloor k/3 \rfloor$-cliques $C$ in $G$ such that the labels of $C$ are strictly greater than those of $A$ and less than those of $B$ and such that $A \cup C$ and $C \cup B$ are both cliques. If it further holds that $(M_2)_{A, B} = 1$, then $A \cup B$ is a clique and $A \cup B \cup C$ is also clique. Therefore the sum output by the algorithm exactly counts the number of triples $(A, B, C)$ such that $A \cup B \cup C$ is a clique, $|A| = |C| = \lfloor k/3 \rfloor$, $|B| = \lceil k/3 \rceil$ and the labels of $C$ are greater than those of $A$ and less than those of $B$. Observe that any clique $\mathcal{C} \in \cl_k(G)$ is counted in this sum exactly once by the triple $(A, B, C)$ where $A$ consists of the lowest $\lfloor k/3 \rfloor$ labels in $\mathcal{C}$, $B$ consists of the highest $\lceil k/3 \rceil$ labels in $\mathcal{C}$ and $C$ contains the remaining vertices of $\mathcal{C}$. Therefore this algorithm solves $\#(k, 2)$-\textsc{clique} as long as Step 1 succeeds.

It suffices to analyze the additional runtime incurred by this post-processing. Observe that the number of cliques output by a call to $\text{greedy-random-sampling}$ with $T$ iterations is at most $T$. Also note that if $\alpha \le \frac{2}{k - 1}$, then $\tau \ge \lfloor \frac{k}{2} \rfloor - 1$. If $k \ge 3$, then it follows that $\tau +1 \ge \lfloor \frac{k}{2} \rfloor \ge \lceil \frac{k}{3} \rceil$. It follows by Theorem \ref{thm:greedysample} that $\max\{ |S_1|, |S_2| \} = \tilde{O}\left( n^{\lceil k/3 \rceil + 1 - \alpha \binom{\lceil k/3 \rceil}{2}} \right)$. Note that computing the matrix $M_P$ takes $\tilde{O}\left( \max\{|S_1|, |S_2|\}^\omega \right) = \tilde{O}\left( n^{\omega \lceil k/3 \rceil + \omega - \omega \alpha \binom{\lceil k/3 \rceil}{2}} \right)$ time. Now observe that all other steps of the algorithm run in $\tilde{O}\left( n^{2\lceil k/3 \rceil - 2\alpha \binom{\lceil k/3 \rceil}{2}} \right)$ time, which completes the proof of the theorem since the matrix multiplication constant satisfies $\omega \ge 2$.
\end{proof}

We remark that for simplicity, we have ignored minor improvements in the runtime that can be achieved by more carefully analyzing Step 4 in terms of rectangular matrix multiplication constants if $k \neq 0 \pmod{3}$. Note that the proof above implicitly used a weak large deviations bound on $|\cl_k(G)|$. More precisely, it used the fact that if $\textsc{greedy-random-sampling}$ with $T$ iterations succeeds, then $|\cl_k(G)| \le T$. Theorem \ref{thm:greedysample} thus implies that $|\cl_k(G)|$ is upper bounded by the minimal settings of $T$ in the theorem statement with probability $1 - n^{-\omega(1)}$ over $G \sim G(n, c, s)$.

When $k \le \tau + 1$, these upper bounds are a $\text{polylog}(n)$ factor from the expectation of $|\cl_k(G)|$. While this was sufficient in the proof of Theorem \ref{thm:matrixaug}, stronger upper bounds will be needed in the next subsection to analyze our deterministic iterative algorithm. The upper tails of $|\cl_k(G)|$ and more generally of the counts of small subhypergraphs in $G(n, c, s)$ have been studied extensively in the literature. We refer to \cite{vu2001large, janson2002infamous, janson2004upper, demarco2012tight} for a survey of the area and recent results. Given a hypergraph $H$, let $N(n, m, H)$ denote the largest number of copies of $H$ that can be constructed in an $s$-uniform hypergraph with at most $n$ vertices and $m$ hyperedges. Define the quantity
$$M_H(n, c) = \max \left\{ m \le \binom{n}{s} : N(n, m, H') \le n^{|V(H')|} c^{|E(H')|} \text{ for all } H' \subseteq H \right\}$$
The following large deviations result from \cite{dudek2010subhypergraph} generalizes a graph large deviations bound from \cite{janson2004upper} to hypergraphs to obtain the following result.

\begin{theorem}[Theorem 4.1 from \cite{dudek2010subhypergraph}] \label{thm:largedev}
For every $s$-uniform hypergraph $H$ and every fixed $\epsilon > 0$, there is a constant $C(\epsilon, H)$ such that for all $n \ge |V(H)|$ and $c \in (0, 1)$, it holds that
$$\mathbb{P}\left[ X_H \ge (1 + \epsilon) \mathbb{E}[X_H] \right] \le \exp\left( - C(\epsilon, H) \cdot M_H(n, c) \right)$$
where $X_H$ is the number of copies of $H$ in $G \sim G(n, c, s)$.
\end{theorem}

Proposition 4.3 in \cite{dudek2010subhypergraph} shows that if $H$ is a $d$-regular $s$-uniform hypergraph and $c \ge n^{-s/d}$ then $M_H(n, c) = \Theta(n^s c^d)$. This implies that
\begin{equation}
    \mathbb{P}\left[ |\cl_k(G)| \ge (1 + \epsilon) \binom{n}{k} c^{\binom{k}{s}} \right] \le \exp\left(-C'(\epsilon,s,k) \cdot n^s c^{\binom{k - 1}{s - 1}} \right) \label{eqn:clique-conc-bound}
\end{equation}
as long as $c \ge n^{-s!(k - s)!/(k - 1)!}$. This provides strong bounds on the upper tails of $|\cl_k(G)|$ that will be useful in the next subsection.

\subsection{Deterministic Iterative Algorithm for Counting in $G(n, c, s)$}

In this section, we present an alternative deterministic algorithm $\textsc{it-gen-cliques}$ achieving a similar runtime to $\textsc{greedy-random-sampling}$. Although they have very different analyses, the algorithm $\textsc{it-gen-cliques}$ can be viewed as a deterministic analogue of $\textsc{greedy-random-sampling}$. Both are constructing cliques one vertex at a time. The algorithm $\textsc{it-gen-cliques}$ takes in cutoffs $C_{s-1}, C_s, \dots, C_k$ and generates sets $S_{s-1}, S_s, \dots, S_k$ as follows:
\begin{enumerate}
\item Initialize $S_{s - 1}$ to be the set of all $(s - 1)$-subsets of $[n]$.
\item Given the set $S_i$, for each vertex $v \in [n]$, iterate through all subsets $A \in S_i$ and add $A \cup \{v \}$ to $S_{i + 1}$ if $A \cup \{v \}$ is a clique and $v$ is larger than the labels of all of the vertices in $A$. Stop if ever $|S_{i+1}| \ge C_{i+1}$.
\item Stop once $S_k$ has been generated and output $S_k$.
\end{enumerate}
Suppose that $C_t$ are chosen to be any high probability upper bounds on the number of $t$-cliques in $G \sim G(n, c, s)$ such as the bounds in Theorem \ref{thm:largedev}. Then we have the following guarantees for the algorithm $\textsc{it-gen-cliques}$.

\begin{theorem}
Suppose that $k$ and $s$ are constants and $c = \Theta(n^{-\alpha})$ for some $\alpha \in (0, 1)$. Let $\tau$ and $\kappa$ be the largest integers satisfying $\alpha \binom{\tau}{s - 1} < 1$ and $\alpha \binom{\kappa}{s - 1} < s$, and let $C_t = 2n^t c^{\binom{t}{s}}$ for each $s \le t \le k$. Then $\textsc{it-gen-cliques}$ with the cutoffs $C_t$ outputs $S_k = \cl_k(G)$ with probability $1 - n^{-\omega(1)}$ where
\begin{enumerate}
\item The runtime of $\textsc{it-gen-cliques}$ is $O\left(n^{\tau+2 - \alpha \binom{\tau + 1}{s}}\right)$ if $\tau + 2 \le k \le \kappa + 1$.
\item The runtime of $\textsc{it-gen-cliques}$ is $O\left(n^{k - \alpha \binom{k - 1}{s}}\right)$ if $k < \tau + 2$.
\end{enumerate}
\end{theorem}

\begin{proof}
Suppose that $k \le \kappa + 1$. We first show that $S_k = \cl_k(G)$ with probability $1 - n^{-\omega(1)}$ in the algorithm $\textsc{it-gen-cliques}$. By a union bound and (\ref{eqn:clique-conc-bound}), it follows that $|\cl_t(G)| < C_t$ for each $s \le t \le k$ with probability at least $1 - (k - s + 1) n^{-\omega(1)}$ since $k \le \kappa + 1$. The following simple induction argument shows that $S_t = \cl_t(G)$ for each $s - 1 \le t \le k$ conditioned on this event. Note that $\cl_{s - 1}(G)$ is by definition the set of all $(s - 1)$-subsets of $[n]$ and thus $S_{s - 1} = \cl_{s - 1}(G)$. If $S_t = \cl_t(G)$, then each $(t + 1)$-clique $\mathcal{C}$ of $G$ is added exactly once to $S_{t + 1}$ as $A \cup \{ v \}$ where $v$ is the vertex of $\mathcal{C}$ with the largest label and $A = \mathcal{C} \backslash \{v \} \in \cl_t(G)$ are the remaining vertices. Now note that the runtime of $\textsc{it-gen-cliques}$ is
$$O\left( \sum_{t = s - 1}^{k - 1} nC_t \right) = O\left( \max_{s - 1 \le t \le k - 1} \left( nC_t \right) \right) = \left\{ \begin{matrix} O\left(n^{\tau+2 - \alpha \binom{\tau + 1}{s}}\right) & \text{if } \tau + 2 \le k \le \kappa + 1 \\ O\left(n^{k - \alpha \binom{k - 1}{s}}\right) & \text{if } k < \tau + 2 \end{matrix} \right.$$
since $k = O(1)$. To see the second inequality, note that $\log_n (C_{t + 1}/C_t) = 1 - \alpha \binom{t}{s - 1} + O(1/\log n)$. This implies that $C_{t + 1} = \Omega(C_t)$ if $t \le \tau$ and $C_t = O(C_{\tau + 1})$ for all $s \le t \le k$. This completes the proof of the theorem.
\end{proof}

We remark that in the case of $k < \tau + 1$, $\textsc{it-gen-cliques}$ attains a small runtime improvement over $\textsc{greedy-random-sampling}$. However, the algorithm $\textsc{greedy-random-sampling}$ can be modified to match this runtime up to a $\text{polylog}(n)$ factor by instead generating the $(k - 1)$-cliques of $G$ and applying the last step of $\textsc{it-gen-cliques}$ to generate the $k$-cliques of $G$. We also remark that $\textsc{it-gen-cliques}$ can also be used instead of $\textsc{greedy-random-sampling}$ in Step 1 of the algorithm in Theorem \ref{thm:matrixaug}, yielding a nearly identical runtime of $\tilde{O}\left( n^{\omega \lceil k/3 \rceil - \omega \alpha \binom{\lceil k/3 \rceil - 1}{2}} \right)$ for $\#(k, 2)$-\textsc{clique} on inputs sampled from $G(n, c)$.

\section{Extensions and Open Problems}\label{sec:extensions}
\label{sec:openproblems}

In this section, we outline several extensions of our methods and problems left open after our work.

\paragraph{Improved Average-Case Lower Bounds} A natural question is if tight average-case lower bounds for $\#(k, s)$\textsc{-clique} can be shown above the $k$-clique percolation threshold when $s \ge 3$ and if the constant $C$ in the exponent of our lower bounds for the graph case of $s = 2$ can be improved from $1$ to $\omega/9$.

\paragraph{Raising Error Tolerance for Average-Case Hardness}

A natural question is if the error tolerance of the worst-case to average-case reductions in Theorems \ref{thm:averagecasehardnesscounting} and \ref{thm:averagecasehardnessparity} can be increased. We remarked in the introduction that for certain choices of $k$, the error tolerance cannot be significantly increased -- for example, when $k = 3 \log_2 n$, the trivial algorithm that outputs $0$ on any graph has subpolynomial error on graphs drawn from $G(n,1/2)$, but is useless for reductions from worst-case graphs. Nevertheless, for other regimes of $k$, such as when $k = O(1)$ is constant, counting $k$-cliques with error probability less than $1/4$ on graphs drawn from $G(n,1/2)$ appears to be nontrivial. It is an open problem to prove hardness for such a regime. In general, one could hope to understand the tight tradeoffs between computation time, error tolerance, $k$, $c$, and $s$ for $k$-clique-counting on $G(n,c,s)$.

\paragraph{Hardness of Approximating Clique Counts}
Another interesting question is if it is hard to approximate the $k$-clique counts, within some additive error $\epsilon$, of hypergraphs drawn from $G(n,c,s)$. Since the number of $k$-cliques in $G(n,c,s)$ concentrates around the mean $\mu \approx c^{\binom{k}{s}} n^k$ with standard deviation $\sigma$, one would have to choose $\epsilon \ll \sigma$ for approximation to be hard.

\paragraph{Inhomogeneous Erd\H{o}s-R\'enyi Hypergraphs}
Consider an inhomogeneous Erd\H{o}s-R\'enyi hypergraph model, where each hyperedge $e$ is independently chosen to be in the hypergraph with probability $c(e)$. 
Also suppose that we may bound $c(e)$ uniformly away from $0$ and $1$ (that is, $c(e) \in [c, 1-c]$ for all possible hyperedges $e$ and for some constant $c$). We would like to prove that \textsc{\#$(k,s)$-clique} and \textsc{Parity-$(k,s)$-clique} are hard on average for inhomogeneous Erd\H{o}s-R\'enyi hypergraphs.
Unfortunately, this does not follow directly from our proof techniques because step 5 in the proof of Theorems \ref{thm:averagecasehardnesscounting} and \ref{thm:averagecasehardnessparity} breaks down due to the inhomogeneity of the model. Nevertheless, steps 1-4 still hold, and therefore we can show that \textsc{\#$(k,s)$-clique} and \textsc{Parity-$(k,s)$-clique} are average-case hard for $k$-partite inhomogeneous Erd\H{o}s-R\'enyi hypergraphs -- when only the edges $e$ that respect the $k$-partition are chosen to be in the hypergraph with inhomogeneous edge-dependent probability $c(e) \in [c, 1-c]$.

\section*{Acknowledgements}
We thank Oded Goldreich and the anonymous reviewers
for helpful feedback that greatly improved the exposition. We also thank Frederic Koehler, Dheeraj Nagaraj, and Austin Stromme for inspiring discussions on related topics.

\printbibliography

\begin{appendix}

\section{Reduction from \textsc{Decide-$(k,s)$-clique} to \textsc{Parity-$(k,s)$-clique}}\label{sec:decidetoparityreduction}
The following is a precise statement and proof of the reduction from \textsc{Decide-$(k,s)$-clique} to \textsc{Parity-$(k,s)$-clique} claimed in Section \ref{sec:worstcasehardnessconjectures}.

\begin{lemma}\label{lem:decidetoparityreduction}
Given an algorithm $A$ for \textsc{Parity-$(k,s)$-clique} with error probability $< 1/3$ on any $s$-uniform hypergraph $G$, there is an algorithm $B$ that runs in time $O(k 2^k |A|)$ and solves \textsc{Decide-$(k,s)$-clique} with error $< 1/3$ on any $s$-uniform hypergraph $G$.
\end{lemma}
\begin{proof}
Let $\mathrm{cl}_k(G)$ denote the set of $k$-cliques in hypergraph $G = (V,E)$. Consider the polynomial $$P_G(x_V) = \sum_{S \in \mathrm{cl}_k(G)} \prod_{v \in S} x_v \pmod{2},$$ over the finite field $\FF_2$. If $G$ has a $k$-clique at vertices $S \subset V$, then $P_G$ is nonzero, because $P_G(1_S) = 1$. If $G$ has no $k$-clique, then $P_G$ is zero. Therefore, deciding whether $G$ has a $k$-clique reduces to testing whether or not $P_G$ is identically zero. $P_G$ is of degree at most $k$, so if $P_G$ is nonzero on at least one input, then it is nonzero on at least a $2^{-k}$ fraction of inputs. One way to see this is that if we evaluate $P_G$ at all points $a \in \{0,1\}^m$, the result is a non-zero Reed-Muller codeword in $RM(k,m)$. Since the distance of the $RM(k,m)$ code is $2^{m-k}$, and the block-length is $2^m$, the claim follows \cite{muller1954application}. We therefore evaluate $P_G$ at $c \cdot 2^k$ independent random inputs for some large enough $c > 0$, accept if any of the evaluations returns 1, and reject if all of the evaluations return 0. Each evaluation corresponds to calculating \textsc{Parity-$(k,s)$-clique} on a hypergraph $G'$ formed from $G$ by removing each vertex independently with probability $1/2$. As usual, we boost the error of $A$ by running the algorithm $O(k)$ times for each evaluation, and using the majority vote.
\end{proof}

\section{Proof of Lemma \ref{lem:reedsolomon}}\label{app:reedsolomonproof}
We restate and prove Lemma \ref{lem:reedsolomon}.
\begin{lemma}[Theorem 4 of \cite{gemmell1992highly}] Let $\FF$ be a finite field with $|\FF| = q$ elements.
Let $N > 0$ and $1 \leq D < q/12$. Let $f : \FF^N \to \FF$ be a polynomial of degree at most $D$. If there is an algorithm $A$ running in time $T(A,N)$ such that $$\PP_{x \sim \mathrm{Unif}\left[\FF^N\right]} [A(x) = f(x)] > 2/3,$$ then there is an algorithm $B$ running in time $O((N+D^2)D \log^2 q + T(A,N) \cdot D)$ such that for {\em any} $x \in \FF^N$, it holds that $\PP[B(x) = f(x)] > 2/3$.
\end{lemma}
\begin{proof}
Our proof of the lemma is based off of the proof that appears in \cite{ball2017average}. The only difference is that in \cite{ball2017average}, the lemma is stated only for finite fields whose size is a prime. Suppose we wish to calculate $f(x)$ for $x \in \FF^N$. In order to do this, choose $y_1,y_2 \stackrel{i.i.d}{\sim} \mathrm{Unif}[\FF^N]$, and define the polynomial $g(t) = x + t y_1 + t^2 y_2$ where $t \in \FF$. We use $A$ to evaluate $f(g(t))$ at $m$ different values $t_1,\ldots,t_m \in \FF$. This takes $O(m N \log^2 q + m \cdot T(A,N))$ time. Suppose without loss of generality that $D \geq 9$. Since $g(t_i)$ and $g(t_j)$ are pairwise independent and uniform in $\FF^N$ for any distinct $t_i,t_j \neq 0$, by the second-moment method, with probability $> 2/3$, at most $(m-2D)/2$ of our evaluations of $f(g(t))$ will be incorrect if we take $m = 12D$. Thus, since $f(g(t))$ is a univariate polynomial of degree at most $2D$, we may use Berlekamp-Welch to recover $f(g(0)) = f(x)$ in $O(m^3)$ arithmetic operations over $\FF$, each of which takes $O(\log^2 q)$ time. 
\end{proof}

\section{Tightness of Bounds in Section~\ref{sec:randombinaryexpansions}}\label{app:binarytightness}
In this appendix, we briefly discuss the tightness of the bounds on $t$ in Lemma~\ref{lem:tvislowfourier} and how the case of $c = 1/2$ differs from $c \neq 1/2$. Note that if $q_i = 1/2$ for each $i$, then $\latestedits{Y = }\sum_{i = 0}^t \latestedits{Z^{(i)}} \cdot 2^i$ is uniformly distributed on $\{0, 1, \dots, 2^{t + 1} - 1 \}$. It follows that \latestedits{the random variable $X \in \FF_p$ defined by $X \equiv Y \pmod{p}$ satisfies}
$$\dtv\left( \cL(\latestedits{X}), \text{Unif}[\mathbb{F}_p] \right) = \sum_{x \in \mathbb{F}_p} \left| p^{-1} - \mathbb{P}[\latestedits{X} = x] \right|_+ = \frac{a(p - a)}{2^{t+1}p} \le \frac{p}{2^{t+1}}$$
if $0 \le a \le p - 1$ is such that $2^{t+1} \equiv a \pmod{p}$. Here, $|\cdot |_+$ denotes $| x |_+ = \max(x, 0)$. Therefore $\latestedits{X}$ is within total variation of $1/\text{poly}(p)$ of $\text{Unif}[\mathbb{F}_p]$ if $t = \Omega(\log p)$. However, note that for $c$ constant and $\epsilon = 1/\text{poly}(p)$, our lemma requires that $t = \Omega(\log^2 p)$. This raises the question: is the additional factor of $\log p$ necessary or an artifact of our analysis? We answer this question with an example suggesting that the extra $\log p$ factor is in fact necessary and that the case $c = 1/2$ is special.

Suppose that $p$ is a Mersenne prime with $p = 2^r - 1$ for some prime $r$ and for simplicity, take $q_i = 1/3$ for each $i$. Observe by the triangle inequality that
$$\left|\hat{f}(1)\right| = \left| \sum_{x \in \mathbb{F}_p} \left(f(x) - p^{-1}\right) \cdot \omega^x \right| \le \left\| f - p^{-1} \cdot \mathbf{1} \right\|_1 = 2 \cdot \dtv\left( \cL(\latestedits{X}), \text{Unif}[\mathbb{F}_p] \right)$$
Now suppose that $t = ar - 1$ for some positive integer $a$. As shown in the lemma, we have
$$\left|\hat{f}(1)\right|^2 = \prod_{i = 0}^t \left| \frac{2}{3} + \frac{1}{3} \cdot \omega^{2^i} \right|^2 = \left[ \prod_{i = 0}^{r - 1}\left( \frac{5}{9} + \frac{4}{9} \cdot \cos\left(\frac{2\pi}{p} \cdot 2^i \right) \right) \right]^a$$
where the second equality is due to the fact that the sequence $2^i$ has period $r$ modulo $p$. Now observe that since $\frac{5}{9} + \frac{4}{9} \cdot \cos(x) \ge e^{-x^2}$, we have that
$$\prod_{i = 0}^{r - 1}\left( \frac{5}{9} + \frac{4}{9} \cdot \cos\left(\frac{2\pi}{p} \cdot 2^i \right) \right) \ge \exp\left( - \frac{4\pi^2}{p^2} \sum_{i = 0}^{r - 1} 2^{2i} \right) = \exp\left( - \frac{4\pi^2}{p^2} \cdot \frac{2^{2r} - 1}{3} \right) = \Omega(1)$$
which implies that $a$ should be $\Omega(r)$ for $\hat{f}(1)$ to be polynomially small in $p$. Thus the extra $\log p$ factor is necessary in this case and our analysis is tight. Note that in the special case of $c = 1/2$, the factors in the expressions for $\hat{f}(s)$ are of the form $\frac{1}{2} + \frac{1}{2} \cdot \omega^{2^i \cdot s}$ which can be arbitrarily close to zero. We remark that the construction, as stated, relies on there being infinitely many Mersenne primes. However, it seems to suggest that the extra $\log p$ factor is necessary. Furthermore, similar examples can be produced with $p$ that are not Mersenne, as long as the order of $2$ modulo $p$ is relatively small.

\section{Clique Counts in Sparse Erd\H{o}s-R\'{e}nyi Hypergraphs}
\label{sec:cliquecounts}
We prove the following classical lemma from Section \ref{subsec:greedysampling}.

\begin{lemma}
For fixed $\alpha \in (0, 1)$ and $s$, let $\kappa \ge s$ be the largest positive integer satisfying $\alpha \binom{\kappa}{s - 1} < s$. If $G \sim G(n, c, s)$ where $c = O(n^{-\alpha})$, then $\mathbb{E}[|\cl_k(G)|] = \binom{n}{k} c^{\binom{k}{s}}$ and $\omega(G) \le \kappa + 1 + t$ with probability at least $1 - O\left(n^{-\alpha t(1 - s^{-1}) \binom{\kappa + 2}{s - 1}}\right)$ for any fixed nonnegative integer $t$, where the constant in the $O(\cdot)$ notation can depend on $t$.
\end{lemma}

\begin{proof}
Let $C > 0$ be such that $c \le Cn^{-\alpha}$ for sufficiently large $n$. For any given set $\{ v_1, v_2, \dots, v_k\}$ of $k$ vertices in $[n]$, the probability that all hyperedges are present among $\{ v_1, v_2, \dots, v_k\}$ and thus these vertices form a $k$-clique in $G$ is $c^{\binom{k}{s}}$. Linearity of expectation implies that the expected number of $k$-cliques is $\mathbb{E}[|\cl_k(G)|] = \binom{n}{k} c^{\binom{k}{s}}$. Now consider taking $k = \kappa + 2 + t$ and note that
\begin{align*}
\mathbb{E}[|\cl_k(G)|] &= \binom{n}{k} c^{\binom{k}{s}} \\
&\le n^k c^{\binom{k}{s}} \le C^{\binom{k}{s}} \cdot \exp\left( \left( 1 - \frac{\alpha}{s} \binom{k - 1}{s - 1} \right) k \log n \right) \\
&\le C^{\binom{k}{s}} \cdot \exp\left( \left( 1 - \frac{\alpha}{s} \binom{\kappa + 1}{s - 1} \right) k \log n - \frac{\alpha}{s} \cdot t \binom{\kappa + 1}{s - 2} k \log n \right) \\
&\le C^{\binom{k}{s}} \cdot \exp\left( - \frac{\alpha}{s} \cdot t \binom{\kappa + 1}{s - 2} k \log n \right) \\
&= C^{\binom{k}{s}} \cdot \exp\left( - \frac{\alpha}{s} \cdot t 
\frac{s-1}{\kappa+2} \binom{\kappa + 2}{s - 1} k \log n \right) \\
&\le C^{\binom{k}{s}} n^{-\alpha t(1 - s^{-1}) \binom{\kappa + 2}{s - 1}}
\end{align*}
where we use $\binom{\kappa + 1 + t}{s - 1} \ge \binom{\kappa + 1}{s - 1} + t \binom{\kappa + 1}{s - 2}$ by iteratively applying Pascal's identity, as well as $\alpha \binom{\kappa+1}{s-1} > s$ and $k \ge \kappa + 2$. Observe that $\kappa = O(1)$ and thus $C^{\binom{k}{s}} = O(1)$. Now by Markov's inequality, it follows that $\mathbb{P}[\omega(G) \ge k] = \mathbb{P}[|\cl_k(G)| \ge 1] \le \mathbb{E}[|\cl_k(G)|]$, completing the proof of the lemma.
\end{proof}

\section{Analysis of Greedy Random Sampling}
\label{sec:grs-analysis}
This section is devoted to proving Theorem \ref{thm:greedysample}, which is restated below for convenience.

\begin{theorem}
Let $k$ and $s$ be constants and $c = \Theta(n^{-\alpha})$ for some $\alpha \in (0, 1)$. Let $\tau$ be the largest integer satisfying $\alpha \binom{\tau}{s - 1} < 1$ and suppose that
$$T \ge \left\{ \begin{matrix} 2n^{\tau + 1} c^{\binom{\tau + 1}{s}} (3 \log n)^{(k - \tau) (1 + \epsilon)} & \text{if } k \ge \tau + 1 \\ 2n^{k} c^{\binom{k}{s}} (\log n)^{1 + \epsilon} & \text{if } k < \tau + 1 \end{matrix} \right.$$
for some $\epsilon > 0$. Then $\textsc{greedy-random-sampling}$ run with $T$ iterations terminates with $S = \cl_k(G)$ with probability $1 -n^{-\omega(1)}$ over the random bits of the algorithm $\textsc{greedy-random-sampling}$ and over the choice of random hypergraph $G \sim G(n, c, s)$.
\end{theorem}

\begin{proof}
We first consider the case where $k \ge \tau + 1$. Fix some $\epsilon > 0$ and let $v = (v_1, v_2, \dots, v_k)$ be an ordered tuple of distinct vertices in $[n]$. Define the random variable
$$Z_v = n(n-1)\cdots (n - s + 2) \prod_{i = s - 1}^{k-1} \left| \textsc{cn}_G(v_1, v_2, \dots, v_i) \right|$$
The key property of $Z_v$ is that, in each iteration of $\textsc{greedy-random-sampling}$, the probability that the $k$ vertices $v_1, v_2, \dots, v_k$ are chosen in that order is exactly $1/Z_v$. The proof of this theorem will proceed by establishing upper bounds on $Z_v$ that hold for all $k$-cliques $v$ with high probability over the randomness of $G$, which will yield a bound on the number of iterations $T$ needed to exhaust all such $k$-cliques in $G$.

Consider the following event over the sampling $G \sim G(n, c, s)$
$$A_v = \left\{ Z_v \ge 2n^{\tau + 1} c^{\binom{\tau + 1}{s}} (3 \log n)^{(k - 1 - \tau) (1 + \epsilon)} \quad \text{and} \quad \{ v_1, v_2, \dots, v_k \} \in \cl_k(G) \right\}$$
We now proceed to bound the probability of $A_v$ through simple Chernoff and union bounds over $G$. In the next part of the argument, we condition on the event that $\{ v_1, v_2, \dots, v_k \}$ forms a clique in $G$. For each $i \in \{s - 1, s, \dots, k - 1\}$, let $Y_{v, i}$ be the number of common neighbors of $v_1, v_2, \dots, v_i$ in $V(G) \backslash \{v_1, v_2, \dots, v_k\}$. Note that $Y_{v, i} \sim \text{Bin}\left(n - k, c^{\binom{i}{s - 1}}\right)$ and that $\left| \textsc{cn}_G(v_1, v_2, \dots, v_i) \right| = k - i + Y_{v, i}$. The standard Chernoff bound for the binomial distribution implies that for all $\delta_i > 0$,
\begin{align*}
&\mathbb{P}\left[ \left| \textsc{cn}_G(v_1, v_2, \dots, v_i) \right| \ge k - i + (1 + \delta_i) (n - k) c^{\binom{i}{s - 1}} \right] \le \exp \left( - \frac{\delta_i^2}{2 + \delta_i} \cdot (n - k) c^{\binom{i}{s - 1}} \right)
\end{align*}
Now define $\kappa_i$ to be
$$\kappa_i = (n - k)^{-1} c^{- \binom{i}{s - 1}} \cdot (\log n)^{1 + \epsilon}$$
for each $i \in \{s-1, s, \dots, k - 1\}$. Let $\delta_i = \sqrt{\kappa_i}$ if $i \le \tau$ and $\delta_i = \kappa_i$ if $i > \tau$. Note that for sufficiently large $n$, $\delta_i < 1$ if $i \le \tau$ and $\delta_i \ge 1$ if $i > \tau$. These choices of $\delta_i$ ensure that the Chernoff upper bounds above are each at most $\exp\left( - \frac{1}{3} (\log n)^{1 + \epsilon} \right)$ for each $i$. A union bound implies that with probability at least $1 - k\exp\left( - \frac{1}{3} (\log n)^{1 + \epsilon} \right)$, it holds that
$$\left| \textsc{cn}_G(v_1, v_2, \dots, v_i) \right| < k - i + (1 + \delta_i) (n - k) c^{\binom{i}{s - 1}} < (1 + 2\delta_i) (n-k) c^{\binom{i}{s - 1}}$$
for all $i$ and sufficiently large $n$. Here, we used the fact that $\delta_i (n - k) c^{\binom{i}{s - 1}} = \omega(1)$ for all $i$ by construction and $k = O(1)$. Observe that $(1 + 2\delta_i) (n-k) c^{\binom{i}{s - 1}} \le 3(\log n)^{1 + \epsilon}$ for all $i \ge \tau + 1$. These inequalities imply that
\begin{align*}
\log Z_v &< \log n^{s - 1} + \sum_{i = s - 1}^\tau \log \left( (1 + 2\delta_i) (n-k) c^{\binom{i}{s - 1}} \right) + (k - 1 - \tau) (1 + \epsilon) \log (3 \log n) \\
&< \log n^{\tau + 1} + (\log c) \sum_{i = s - 1}^\tau \binom{i}{s - 1} + \sum_{i = s - 1}^\tau \log (1 + 2\delta_i) + (k - 1 - \tau) (1 + \epsilon) \log (3 \log n) \\
&\le \log \left( n^{\tau + 1} c^{\binom{\tau + 1}{s}} \right) + (k - 1 - \tau) (1 + \epsilon) \log (3 \log n) + 2 \sum_{i = s - 1}^\tau \delta_i \\
&\le \log \left( n^{\tau + 1} c^{\binom{\tau + 1}{s}} \right) + (k - 1 - \tau) (1 + \epsilon) \log (3 \log n) + o(1)
\end{align*}
The last inequality holds since $\tau = O(1)$ and since $\delta_i \lesssim (\log n)^{\frac{1}{2} + \frac{\epsilon}{2}} n^{-\frac{1}{2} + \frac{1}{2}\alpha \binom{\tau}{s - 1}} = o(1)$ for all $i \le \tau$ because of the definition that $\alpha \binom{\tau}{s - 1} < 1$. In summary, we have shown that for sufficiently large $n$
\begin{align*}
&\mathbb{P}\left[ Z_v \ge 2n^{\tau + 1} c^{\binom{\tau + 1}{s}} (3 \log n)^{(k - 1 - \tau) (1 + \epsilon)} \, \Big| \, \{ v_1, v_2, \dots, v_k \} \in \cl_k(G) \right] \\
&\quad \quad \le k\exp\left( - \frac{1}{3} (\log n)^{1 + \epsilon} \right) = n^{-\omega(1)}
\end{align*}
for any $k$-tuple of vertices $v = (v_1, v_2, \dots, v_k)$. Since $\mathbb{P}\left[ \{ v_1, v_2, \dots, v_k \} \in \cl_k(G) \right] = c^{\binom{k}{s}}$, we have that $\mathbb{P}[A_v] \le c^{\binom{k}{s}} n^{-\omega(1)} = n^{-\omega(1)}$ for each $k$-tuple $v$. Now consider the event
\begin{align*}
B &= \Big\{ Z_v < 2n^{\tau + 1} c^{\binom{\tau + 1}{s}} (3 \log n)^{(k - 1 - \tau) (1 + \epsilon)} \text{ for all } v  \text{ such that } \{ v_1, v_2, \dots, v_k \} \in \cl_k(G) \Big\}
\end{align*}
Note that $\overline{B} = \bigcup_{k\text{-tuples } v} A_v$ and a union bound implies that $\mathbb{P}[B] \ge 1 - \sum_{v} \mathbb{P}[A_v] \ge 1 - n^k \cdot n^{-\omega(1)} = 1 - n^{-\omega(1)}$ since there are fewer than $n^k$ $k$-tuples $v$.

We now show that as long as $B$ holds over the random choice of $G$, then the algorithm $\textsc{greedy-random-sampling}$ terminates with $S = \cl_k(G)$ with probability $1 - n^{-\omega(1)}$ over the random bits of $\textsc{greedy-random-sampling}$, which completes the proof of the lemma in the case $k > \tau + 1$. In the next part of the argument, we consider $G$ conditioned on the event $B$. Fix some ordering $v = (v_1, v_2, \dots, v_k)$ of some $k$-clique $\mathcal{C} = \{ v_1, v_2, \dots, v_k\}$ in $G$. Recall that in any one of the $T$ iterations of $\textsc{greedy-random-sampling}$, the probability that the $k$ vertices $v_1, v_2, \dots, v_k$ are chosen in that order is exactly $1/Z_v$. Since the $T$ iterations of $\textsc{greedy-random-sampling}$ are independent, we have that
$$\mathbb{P}\left[v \text{ is never chosen in a round} \right] = \left( 1 - \frac{1}{Z_v} \right)^T \le \exp\left( - \frac{T}{Z_v} \right) = n^{-\omega(1)}$$
since $T$ is chosen so that $T \ge Z_v (\log n)^{3(1+\epsilon)}$ for all $k$-tuples $v$, given the event $B$. Since there are at most $n^k$ possible $v$, a union bound implies that every such $v$ is chosen in a round of $\textsc{greedy-random-sampling}$ with probability at least $1 - n^{k} \cdot n^{- \omega(1)} = 1 - n^{-\omega(1)}$ over the random bits of the algorithm. In this case, $S = \cl_k(G)$ after the $T$ rounds of $\textsc{greedy-random-sampling}$. This completes the proof of the theorem in the case $k \ge \tau + 1$.

We now handle the case $k < \tau + 1$ through a nearly identical argument. Define $\kappa_i$ as in the previous case and set $\delta_i = \sqrt{\kappa_i}$ for all $i \in \{s - 1, s, \dots, k - 1\}$. By the same argument, for each $k$-tuple $v$ we have with probability $1 - n^{-\omega(1)}$ over the choice of $G$ that
\begin{align*}
\log Z_v &< \log n^{s - 1} + \sum_{i = s - 1}^{k - 1} \log \left( (1 + 2\delta_i) (n-k) c^{\binom{i}{s - 1}} \right) \\
&< \log n^k + (\log c) \sum_{i = s - 1}^{k - 1} \binom{i}{s - 1} + 2\sum_{i = s - 1}^{k - 1} \delta_i \\
&= \log \left( n^k c^{\binom{k}{s}} \right) + o(1)
\end{align*}
where again $\delta_i \lesssim (\log n)^{\frac{1}{2} + \frac{\epsilon}{2}} n^{-\frac{1}{2} + \frac{1}{2}\alpha \binom{\tau}{s - 1}} = o(1)$ for all $i \le k - 1 < \tau$. Define the event
$$B' = \left\{ Z_v < 2n^k c^{\binom{k}{s}} \text{ for all } v \text{ such that } \{ v_1, v_2, \dots, v_k \} \in \cl_k(G) \right\}$$
Note that $T$ is such that $T \ge Z_v (\log n)^{1 + \epsilon}$ for all $v$ if $B'$ holds. Now repeating the rest of the argument from the $k \ge \tau + 1$ case shows that $\mathbb{P}[B'] \ge 1 - n^{-\omega(1)}$ and that $\textsc{greedy-random-sampling}$ terminates with $S = \cl_k(G)$ with probability $1 - n^{-\omega(1)}$ over its random bits if $G$ is such that $B'$ holds. This completes the proof of the theorem.
\end{proof}

\end{appendix}

\end{document}